\documentclass[lettersize,journal,draftclsnofoot,onecolumn,12pt]{IEEEtran}
\usepackage{enumerate}
\usepackage{cite}
\usepackage{graphicx}
\usepackage{epstopdf}
\usepackage[utf8]{inputenc}
\usepackage{amsmath}
\usepackage{amsthm}
\usepackage{amssymb}
\usepackage{bbm}
\usepackage{mathtools}
\usepackage{mathdots}
\usepackage{comment}
\usepackage{textcomp}
\usepackage{url}
\usepackage{xcolor}
\usepackage{caption}
\usepackage{subcaption}
\usepackage{multicol}

\providecommand{\keywords}[1]
{
  \small	
  \textbf{\textit{Keywords---}} #1
}

\usepackage[colorinlistoftodos]{todonotes}
\usepackage{hyperref}
\hypersetup{
    colorlinks=true,
    linkcolor=blue,
    filecolor=magenta,      
    urlcolor=cyan,
    pdftitle={Overleaf Example},
    pdfpagemode=FullScreen,
    }

\DeclareMathOperator{\rank}{\mathbf{Rank}}
\DeclareMathOperator{\reals}{\mathbb{R}}
\def\norm#1{\left\lVert #1\right\rVert}

\usepackage{algorithm}
\usepackage[noend]{algpseudocode}

\let\oldnl\nl
\newcommand{\nonl}{\renewcommand{\nl}{\let\nl\oldnl}}

\usepackage{array}

\usepackage{stfloats}

\newtheorem{theorem}{Theorem}

\newtheorem{proposition}{Proposition}

\newtheorem{definition}{Definition}
\newtheorem{remark}{Remark}

\DeclareMathOperator*{\argmin}{argmin}
\DeclarePairedDelimiter\ceil{\lceil}{\rceil}

\begin{document}

\title{\LARGE \bf Lp Quasi-norm Minimization: Algorithm and Applications}
\author{Omar M.Sleem$^\star$ \quad M.E. Ashour$^{\ddag}$  \\ N.S. Aybat$^\dagger \S$ \quad  Constantino M. Lagoa$^{\star}$  \\ 
\small $^\star$Dep. of Electrical Engineering, Pennsylvania State University, State College, PA 16801, USA. \\	
\small $^\ddag$Wireless R$\&$D Department, Qualcomm Technologies, Inc, San Diego, CA 92121, USA.\\
\small $^\dagger$Department of Industrial and Manufacturing Engineering, Pennsylvania State University, State College, PA 16801, USA.\\
\small Email: \{\ oms46@psu.edu, mashour@qti.qualcomm.com, nsa10@psu.edu, cml18@psu.edu \}\
}
\date{}
\maketitle
\begin{abstract}
Sparsity finds applications in areas as diverse as statistics, machine learning, and signal processing. Computations over sparse structures are less complex compared to their dense counterparts, and their storage consumes less space. This paper proposes a heuristic method for retrieving sparse approximate solutions of optimization problems via minimizing the $\ell_{p}$ quasi-norm, where $0<p<1$. An iterative two-block ADMM algorithm for minimizing the $\ell_{p}$ quasi-norm subject to convex constraints is proposed. For $p=s/q<1$, $s,q \in \mathbb{Z}_{+}$, the proposed algorithm requires solving for the roots of a scalar degree $2q$ polynomial as opposed to applying a soft thresholding operator in the case of $\ell_{1}$.
The merit of that algorithm relies on its ability to solve the $\ell_{p}$ quasi-norm minimization subject to any convex set of constraints. However, it suffers from low speed, due to a convex projection step in each iteration, and the lack of mathematical convergence guarantee. We then aim to vanquish these shortcomings by relaxing the assumption on the constraints set to be the set formed due to convex and differentiable, with Lipschitz continuous gradient, functions, i.e. specifically, polytope sets. Using a proximal gradient step, we mitigate the convex projection step and hence enhance the algorithm speed while proving its convergence.
We then present various applications where the proposed algorithm excels, namely, matrix rank minimization, sparse signal reconstruction from noisy measurements, sparse binary classification, and system identification. The results demonstrate the significant gains obtained by the proposed algorithm compared to those via $\ell_{1}$ minimization.
\end{abstract}
\vspace{-0.5cm}
\section{Introduction}
\subsection{Motivation}
In numerical analysis and scientific computing, a sparse matrix/array is the one with 
many of its elements being zeros. The number of zeros divided by the total number of elements is called sparsity. Sparse data is often easier to store and process. Hence, techniques for deriving sparse solutions and exploiting them have attracted the attention of many researchers in various engineering fields like machine learning, signal processing, and control theory.

The taxonomy of sparsity can be studied through the rank minimization problem (RMP). It has been lately considered in many engineering applications including control design and system identification. This is because the notions of complexity and system order can be closely related to the matrix rank. The RMP can be formulated as follows:
\begin{equation} \label{1}
\begin{aligned}
\min_{\mathbf{X}} \quad & \rank(\mathbf{X}), \quad
\textrm{s.t.}  & \mathbf{X} \in \mathcal{M},\\
\end{aligned}
\end{equation}
where $\mathbf{X} \in \mathbb{R}^{m\times n}$ and $\mathcal{M}\subset\reals^{m\times n}$ is a convex set. The problem \eqref{1} in its generality is NP-hard \cite{2}. Therefore, 
polynomial time algorithms for solving large-scale problems of the form in \eqref{1} are not currently known. 
Hence, currently adopted methods for solving such problems are approximate and structured heuristics. 

A special case of RMP is the sparse vector recovery (SVR) problem involving $\ell_{0}$ pseudo-norm minimization given by: 
\begin{equation} \label{zero norm}
\begin{aligned}
\min_{\mathbf{x}} \quad &  \left\lVert \mathbf{x}\right\rVert_{0}, \quad
\textrm{s.t.} & \mathbf{x} \in \mathcal{V},\\
\end{aligned}
\end{equation}
where $\mathbf{x} \in \mathbb{R}^{n}$, $\mathcal{V}\subset\reals^n$ is a closed convex set and $\norm{\cdot}_{0}$ counts the number of the non-zero 
elements of its argument. From the definition of the rank being the number of non-zero singular values of a matrix, it can be easily realized that \eqref{1} is a generalized form of \eqref{zero norm}. 

Various works -- which will be discussed in the next section in more detail -- have explored efficient solution techniques for the problems in \eqref{1} and \eqref{zero norm} independently using Schatten-p and $\ell_{p}$ quasi-norm relaxations respectively. However, all of these methods either assume a specific structure for the convex set $\mathcal{M}$ in \eqref{1} or only work for the special case in \eqref{zero norm}; hence, they lack generality. Using the fact that the Schatten-p quasi-norm is the $\ell_{p}$ quasi-norm of the matrix singular values, we aim to 
design an efficient heuristic method based on Schatten-p relaxation for solving
both problems in a unified manner.
To achieve this, we first 
propose an algorithm for solving $\ell_{p}$ quasi-norm relaxation of the SVR problem in \eqref{zero norm}, and next,
exploiting the fact that \eqref{zero norm} is a special case of \eqref{1}, we then 
employ the derived $\ell_{p}$ quasi-norm minimization algorithm 
as a building block for the desired generalized algorithm for the rank minimization problem.
\vspace{-0.3cm}
\subsection{Related work}
\subsubsection{Sparse Vector Recovery} \label{SVR}
As discussed, a sparse solution is defined as the one 
having the minimum number of non-zero components while satisfying certain constraints such as a system of linear equations ~\cite{vec1}. Since many signals are either sparse or compressible, SVR problem has found applications in object recognition, classification and compressed sensing problems, see, e.g.,\cite{vec2,vec3,vec4}. In \cite{vec5}, the authors discussed the concept of the sparse representation of signals and systems, where they reviewed the theoretical and empirical results 
on sparse optimization, and discussed sufficient conditions needed for 
uniqueness, stability and computational practicability. 
Different applications for the SVR problem are explored in \cite{vec5} and 
it is argued that in certain denoising and compression tasks, 
the methods for sparse optimization provide state of the art solutions. 

The problem of constructing a sparse solution to undetermined linear systems has received great attention. In \cite{vec6}, the authors surveyed the 
existing algorithms for sparse approximation, namely; 
greedy methods~\cite{vec7,vec13}, the methods based on convex relaxation \cite{vec3,vec4,vec8,vec14}, those involving non-convex optimization \cite{vec9,vec15}, Bayesian framework~\cite{vec10,vec11}, and requiring brute force~\cite{vec12}. They also discussed the computational requirements of each algorithm and their relation 
to each other. 

Sparse optimization problems of the form $\min\{f(\mathbf{x}) + \mu g(\mathbf{x}) \}$ have been extensively studied in the literature, where $g(\mathbf{x})$ is a sparsity inducing function, e.g., $\ell_1$-norm, $f$ is a loss function on measurement errors, e.g., $f(\mathbf{x})=\norm{A \mathbf{x}-b}^2$, and $\mu>0$ is a trade-off parameter between data-fidelity and sparsity. In \cite{vec1}, the authors considered 
sparse recovery problem from a set of corrupted measurements.
For $g(\cdot)=\norm{\cdot}_1$, they established a sufficient condition for
the \emph{exact} sparse signal recovery, 
i.e., restricted isometry property (RIP). 
Motivated by the fact that $\|\mathbf{x}\|_{p}^{p} \to \|\mathbf{x}\|_{0}$ as $p \to 0$, it is natural to consider the above problem with $g$ set to $\ell_{p}$ quasi-norm for $p\in(0,1)$. Hence, the authors in \cite{vec16} presented theoretical results demonstrating the ability of the $\ell_{p}$ quasi-norm to recover sparse signals from noisy measurements. Under more relaxed RIP conditions, they showed that the $\ell_{p}$ quasi-norm provides better theoretical guarantees in terms of stability and robustness than the $\ell_{1}$ minimization. In \cite{vec9}, the problem of SVR via the $\ell_{p}$ quasi-norm minimization 
from small number of linear measurements of the target signal was considered. This setting is important in applications where data acquisition is difficult or expensive. However, 
the proposed approach in~\cite{vec9} 
has limited applicability due to its long reconstruction time compared to the $\ell_{1}$ norm. In \cite{vec17}, the authors exploited Fourier-based algorithms for convex optimization to solve 
sparse signals reconstruction problem via the $\ell_{p}$ quasi-norm minimization. They showed that their approach combines the construction abilities of the non-convex methods with the speed of the convex ones. 
In \cite{vec18} the authors proposed an approach, for sparse reconstruction, replacing the non-convex function with a quadratic convex one. In \cite{ref1}, an alternating direction method of multipliers (ADMM) based algorithm that enforces both sparsity and group sparsity using non-convex regularization is presented. An iterative half thresholding algorithm for fast solution of the $\ell_{0.5}$ regularization is proposed in \cite{ref6}. The authors proved the existence of the resolvent of gradient of $||\mathbf{x}||_{0.5}^{0.5}$ , calculated its analytic expression, and 
derived a thresholding representation of solutions for $\ell_{0.5}$ regularization. In \cite{ref7}, the convergence of the iterative half thresholding algorithm is studied, where it was shown that, under certain conditions, the half thresholding algorithm converges, to a local minimizer of the regularized problem, with a linear convergence rate.
Conditions for the convergence of an ADMM algorithm that solves the problem of minimizing the sum of a smooth function with a bounded Hessian and a non-smooth function are derived in \cite{ref4}. In \cite{ref3}, the convergence of ADMM for minimizing a non-convex and possibly non-smooth objective function subject to equality constraints is analyzed. The developed convergence guarantee covers a variety of non-convex objectives including piece-wise linear functions, $\ell_{p}$ quasi-norm and Schatten-$p$ quasi-norm ($0<p<1$), while allowing non-convex constraints as well. 
\subsubsection{Rank minimization} \label{RM}
In \cite{3}, the authors aimed 
at determining the least order dynamic output feedback, using same formulation as in (\ref{1}), which stabilizes a given linear time invariant 
system. They found that minimizing the trace instead of the rank results in a Semi-Definite Program (SDP) that can be solved efficiently. However, their solution was only 
applicable for symmetric and square matrices. In \cite{4}, a generalization 
of the latter approach was introduced which is based on replacing the rank in the objective function with the summation of the singular values of the matrix, i.e., the nuclear norm. They showed that this leads to the convex envelope of the non-convex rank objective and boils down to the original trace heuristic when the decision matrix is a symmetric positive semi-definite~(PSD) matrix. Finally, effectiveness of the approach was shown using a frequency domain system identification problem. In \cite{5}, another heuristic based on the logarithm of the determinant was presented as a surrogate for the rank minimization 
over the subspace of PSD matrices, and the authors showed that this formulation can be solved using 
a sequence of trace minimization problems. The authors also extended their heuristic 
to handle matrices that are not necessarily PSD. 
In \cite{1}, the authors 
also studied the existing trace and log determinant heuristics 
for approximating (\ref{1}). They 
discussed the applications of these heuristics for computing a low-rank approximation of 
1) \emph{covariance matrices} for a given dataset so that one can obtain a simple data model, easy to interpret, which is especially important in statistics and signal processing;
2) \emph{Hankel matrices} 
arising in system identification of a time invariant, low-order system for given output realizations; and 3) matrices appearing in various other problems including $H_{\infty}$ and reduced-order $\mu$-synthesis with constant scaling and problems with inertia constraints.

Although the nuclear norm is 
the tightest convex substitute for the non-convex rank function, 
one of its major shortcomings is that it treats all the singular values equally 
in order to be able to preserve the convexity. Therefore, this restricts its performance in applications where the singular values need to be treated differently, e.g., particularly in image denoising. In \cite{6}, the authors proposed an iterative re-weighted nuclear norm heuristic to 
avoid this problem and analyzed its convergence. They also 
proposed a gradient-based algorithm 
and applied it to 
a low-order system identification problem. Experimental results showed that the re-weighted nuclear norm 
leads to a lower order model than the nuclear norm itself. In \cite{7}, the solution of weighted nuclear norm (WNN) problem 
was analyzed under different circumstances where the weights could be in a non-ascending, arbitrary or a non-descending order. The authors 
applied their proposed WNN algorithm to an image denoising 
problem by exploiting the image non-local self-similarity. Numerical results showed that the proposed WNN algorithm outperforms many of the state of the art denoising methods in terms of both quantitative measures and visual quality. 

Another method, inspired by the success of the $\ell_{p}$ quasi-norm $(0<p<1)$ for sparse signal reconstruction, is to enforce 
low-rank 
structure by using the Schatten-p quasi-norm, which is defined as the $\ell_{p}$ quasi-norm of the singular values. In \cite{8}, the authors considered the matrix completion problem, which deals with 
constructing a low-rank matrix, given a subset of its entries. Instead of 
minimizing the nuclear norm, 
the authors proposed a Schatten-p quasi-norm formulation, for which 
they came up with an algorithm and
studied its convergence properties. In each iteration, the sub-problem that needs to be solved 
has a closed-form solution,
which makes it fast and suitable for large-scale problems. To improve the robustness of the solutions to matrix completion problem, in~\cite{10} Schatten-$p$ quasi-norm for low-rank recovery was combined with $\ell_{p}$ quasi-norm $(0 < p \leq 1)$ of the prediction errors on the observed entries.  
The authors proposed an algorithm based on the ADMM, 
which performed better in their numerical experiments than other completion methods like fixed-point continuation and accelerated proximal gradient singular-value thresholding. In \cite{9}, the authors extended the theoretical recovery results previously developped for the nuclear norm 
to Schatten-$p$ quasi-norm using a weaker version of the RIP assumption; they showed that the minimum rank solution can be recovered 
by solving the Schatten-$p$ quasi-norm minimization problem. In \cite{ref8}, the authors developed an iterative re-weighted least squares algorithm to solve an unconstrained $\ell_{p}$ minimization problem. The algorithm, and analysis, are extended to include the low rank recovery problem. 
Another non-convex approaches for matrix optimization problems involving sparsity are developed, by means of a generalized shrinkage operation, in \cite{ref2}. These approaches are applied to the decomposition of video into low rank
and sparse components, which is able to separate moving objects
from the stationary background better than in the convex case.
\vspace{-0.3cm}
\subsection{Contributions} \label{contributions}
Despite the good performance of the various algorithms discussed in \ref{SVR} and \ref{RM} for solving different relaxations of \eqref{zero norm} and \eqref{1}, they are all based on a specific structure of the considered convex constraint set; 
therefore, they are problem specific and lack generality. 
In this work, we present a general-purpose method based on projections onto the constraint set, for which we assume only closed convexity as the specific structure.  \cite{Proj1} and \cite{Proj2} examine the characteristics of the projection on the constraint set. In the latter, the projection technique of every given point on $\ell_{p}$ balls is assumed to be known a priori, whereas in the former, the issue is solved while omitting an important coupling condition for the polynomial equations.

First, 
based on our 
previous work in \cite{14}, we propose an ADMM algorithm (pQN-ADMM) to solve the $\ell_{p}$ quasi-norm relaxation of (\ref{zero norm}). At each iteration, the bottleneck operation is 
to compute 
Euclidean projections on to some particular convex and non-convex sets. 
The proposed algorithm possesses two important properties: 1) Its computational complexity is similar to $\ell_{1}$ minimization algorithms except for the additional effort of solving for the roots of a 
polynomial; 2) No specific structure for the convex set is required. Numerical results, using a SVR and binary classification examples, are presented to show the 
competitive performance of our algorithm 
against the $\ell_{1}$ minimization approach. We then extend 
the proposed algorithm to solve the relaxation of (\ref{1}) 
based on the Schatten-$p$ quasi-norm -- here, we exploit the fact that minimizing the $\ell_{p}$-norm of the vector of singular values is equivalent to minimizing the Schatten-$p$ quasi-norm. 
We consider two different numerical examples: 
1) We formulate time domain system identification problem for minimum order system detection and solve it using the pQN-ADMM approach 
and compare our recovery results against the nuclear norm minimization 
approach of \cite{4}; 2) We consider a matrix completion problem, where the goal is to recover an unknown low-rank matrix based on a small fraction of observed entries. 
Numerical results in both examples show 
that our method is competitive against some of the state-of-the-art algorithms in terms of both the detected system order (for the system identification problem) and the rank of the matrix recovered (for the matrix completion problem). 

Finally, since the derived algorithm depends on a computationally expensive convex projection step in every iteration, we aim to develop a faster algorithm with a mathematical convergence guarantee. Considering only a subset of problems where the constraint set is a polytope, we utilize concepts from the proximal gradient (PG) method to derive a fast algorithm and prove that it converges with a rate $O(\frac{1}{K})$, where $K$ is the iteration budget given to the algorithm.

\section{Notations and basic definitions}
Unless otherwise specified, we denote vectors with lowercase boldface letters, i.e., $\mathbf{x}$, with $i$-th entry as $x_{i}$, while matrices are in uppercase, i.e. $\mathbf{X}$, with $(i,j)$-th entry as $x_{i,j}$. For an integer $n \in \mathbb{Z}_{+}$, $[n]\stackrel{\Delta}{=}\{1,\hdots,n\}$. $\mathbf{1}$ represents a vector of all entries equal to 1, while $\mathbbm{1}_{\mathcal{G}}(.)$ is an indicator function to the set $\mathcal{G}$, i.e., it evaluates to zero if its argument belongs to the set 
$\mathcal{G}$ and is $+\infty$ otherwise. 

For a vector $\mathbf{x}\in\mathbb{R}^{n}$, the general $\ell_{p}$ norm is defined as 
\begin{equation} \label{normdef}
    \left\lVert \mathbf{x}\right\rVert_{p}\stackrel{\Delta}{=}(\sum_{i\in [n]}|x_{i}|^{p})^{\frac{1}{p}}.
\end{equation}
For convenience, we let $\left\lVert \mathbf{x} \right \rVert$ be the well known euclidean norm, i.e., $p=2$. When $0<p<1$, the expression in \eqref{normdef} is termed as the quasi-norm satisfying the same axioms of the norm except the triangular inequality making it a non-convex function. 
\begin{definition}
Let $\mathbf{X}:\mathbf{V}\longrightarrow\mathbf{W}$ be a linear operator between two normed spaces equipped with $\ell_{p}$ norm, $p\in [1,\infty)$. The induced $p$-norm is defined as,
\begin{equation} \label{inducedpnorm}
    \left\lVert\mathbf{X}\right\rVert_{p}\stackrel{\Delta}{=}\sup_{v\neq 0}\big\{\frac{\left\lVert\mathbf{X}v\right\rVert_{p}}{\left\lVert v\right\rVert_{p}} \big\}.
\end{equation}
\end{definition}
A special case of \eqref{inducedpnorm} is when $p=2$, known as the spectral radius, which can be shown to be the square root of the maximum eigen value of $\mathbf{X}^{\mathrm{H}}\mathbf{X}$, where $\mathbf{X}^{\mathrm{H}}$ is the complex conjugate of the transpose of $\mathbf{X}$, i.e., $\mathbf{X}^{\top}$. In the rest of the analysis, we will drop the subscript 2 in the spectral norm notation and only refer to it with $\left\lVert.\right\rVert$.

For a matrix $\mathbf{X}\in\mathbb{R}^{m\times n}$, the $L_{p,q}$ entry-wise norm is defined as,
\begin{equation} \label{entrywisenorm}
  \left\lVert\mathbf{X}\right\rVert_{p,q}\stackrel{\Delta}{=}(\sum_{j\in [n]}(\sum_{i \in [m]}|x_{ij}|^{p})^{\frac{q}{p}})^{\frac{1}{q}}.
\end{equation}
A special case of \eqref{entrywisenorm} is when $p=q=2$, known as the Frobenius norm, which were refer to by $\left\lVert.\right\rVert_{\mathrm{f}}$.
\begin{definition}
Let $\mathbf{H}_{1} \subset \mathbb{R}^{n}$ and $\mathbf{H}_{2}\subset \mathbb{R}^{m}$ be two separable Hilbert spaces and $\mathbf{X}\in \mathbb{R}^{m\times n}$ be a linear compact operator from $\mathbf{H}_{1}$ to $\mathbf{H}_{2}$, the Schatten-p norm of $\mathbf{X}$ is then defined as,
\begin{equation} \label{schattenpnormdef}
    \left\lVert\mathbf{X}\right\rVert_{p}\stackrel{\Delta}{=}(\!\!\!\sum_{i\in [\min\{m,n\}]}\!\!\!\sigma_{i}(\mathbf{X})^{p})^{\frac{1}{p}},
\end{equation}
where $\sigma_{i}(\mathbf{X})$ is the $i$-th singular value of the matrix $\mathbf{X}$. 
\end{definition}

When $p=1$, equation \eqref{schattenpnormdef} yields to the nuclear norm which is the convex envelope of the rank function. Throughout the paper, we consider a non-convex relaxation for the rank function, specifically $p=1/2$, and compare its performance with the nuclear norm case in the results section.

We define $\ceil*{\cdot}$ as the ceiling operator, $\mathrm{vec}(\mathbf{X}) \in \mathbb{R}^{mn}$ as a vector formed by stacking the columns of the matrix $\mathbf{X}\in\mathbb{R}^{m\times n}$ and $Hankel(.)$ as an operator that outputs a Hankel matrix constructed from the applied vector arguments.
\vspace{-0.3cm}
\section{ Sparse Vector Recovery Algorithm} \label{quasi-norm}
\subsection{Problem Formulation}
This section develops a method for approximating the solution of \eqref{zero norm} using the following relaxation,
\begin{equation} \label{motivProb}
\begin{aligned}
\min_{\mathbf{x}} \quad &  \left\lVert \mathbf{x}\right\rVert_{p}^{p}, \quad
\textrm{s.t.} & \mathbf{x} \in \mathcal{V},\\
\end{aligned}
\end{equation}
where $p \in (0, 1]$ and $\mathcal{V}$ is a closed convex set. Problem \eqref{motivProb} is convex 
for $p=1$; hence, can be solved to optimality efficiently. 
However, the problem becomes non-convex when $p<1$. An epigraph equivalent formulation of (\ref{motivProb}) is obtained by introducing the variable $\mathbf{t}=[t_{i}]_{i \in [n]}$:
\begin{equation}\label{lp_epigraph} 
\begin{aligned} 
& \underset{\mathbf{x},\mathbf{t}}{\min}
& & \mathbf{1}^{\top} \mathbf{t},\\
& \mathrm{s.t.}
& & t_{i} \geq |x_{i}|^{p}, \quad i \in [n], \quad \mathbf{x} \in \mathcal{V}.
\end{aligned}
\end{equation}
Let $\mathcal{X} \subset \mathbb{R}^{2}$ denote the epigraph of the scalar function $|x|^{p}$, i.e., $\mathcal{X}=\{(x,t) \in \mathbb{R}^2: t \geq |x|^{p}\}$, which is a non-convex set for $p<1$. Then, (\ref{lp_epigraph}) can be cast as
\begin{equation}\label{lp_indicators}  
\begin{aligned} 
& \underset{\mathbf{x},\mathbf{t}}{\min}
& & \sum_{i \in [n]} \mathbbm{1}_{\mathcal{X}}(x_{i},t_{i}) + \mathbf{1}^{\top}\mathbf{t},
& \mathrm{s.t.}
& &  \mathbf{x} \in \mathcal{V}.
\end{aligned}
\end{equation}
ADMM exploits the structure of the problem to split the optimization over the variables via iteratively solving fairly simple subproblems. In particular, 
we introduce auxiliary variables $\mathbf{y}=[y_{i}]_{i \in [n]}$ and $\mathbf{z}=[z_{i}]_{i \in [n]}$ and obtain an ADMM equivalent formulation of (\ref{lp_indicators}) given by:
\begin{equation}\label{lp_ADMM} 
\begin{aligned} 
& \underset{\mathbf{x},\mathbf{t},\mathbf{y},\mathbf{z}}{\min}
& & \sum_{i \in [n]} \mathbbm{1}_{\mathcal{X}}(x_{i},t_{i}) + \mathbbm{1}_{\mathcal{Y}}(\mathbf{y})+\mathbf{1}^{\top}\mathbf{z}, \\
& \mathrm{s.t.}
& & \mathbf{x}=\mathbf{y}: \boldsymbol\lambda, \quad
 \mathbf{t}=\mathbf{z}: \boldsymbol\theta,
\end{aligned}
\end{equation}
where $\mathcal{Y}$ is the $0$-sublevel set of $f$, i.e., $\mathcal{Y}=\{\mathbf{y} \in \mathbb{R}^{n} : \textbf{y} \in \mathcal{V}\}$. 
The dual variables associated with the constraints $\mathbf{x}=\mathbf{y}$ and $\mathbf{t}=\mathbf{z}$ are $\boldsymbol\lambda$ and 
$\boldsymbol\theta$, respectively. Hence, the Lagrangian function corresponding to (\ref{lp_ADMM}) augmented with a quadratic penalty on the violation of the equality constraints with penalty parameter $\rho > 0$, is given by:
\begin{align}\label{Lagrangian}
&\mathcal{L}_{\rho}(\mathbf{x},\mathbf{t},\mathbf{y},\mathbf{z},\boldsymbol\lambda,\boldsymbol\theta)=
\sum_{i \in [n]} \!\mathbbm{1}_{\mathcal{X}}(x_{i},t_{i}) + \mathbbm{1}_{\mathcal{Y}}(\mathbf{y})+\mathbf{1}^{\top}\mathbf{z} \notag \\
& + \boldsymbol\lambda^{\top}\!(\mathbf{x}-\mathbf{y}) 
+ \boldsymbol\theta^{\top}\!(\mathbf{t}-\mathbf{z}) + \frac{\rho}{2}\! \left(\|\mathbf{x}-\mathbf{y}\|^{2} +\|\mathbf{t}-\mathbf{z}\|^{2} \right)\!. 
\end{align}
Considering the two block variables $(\mathbf{x},\mathbf{t})$ and $(\mathbf{y},\mathbf{z})$, ADMM \cite{11} consists of the following iterations:
\begin{eqnarray}
(\mathbf{x},\mathbf{t})^{k+1} &\! = &\! \underset{\mathbf{x},\mathbf{t}}{\mathrm{argmin}}~
\mathcal{L}_{\rho}(\mathbf{x},\mathbf{t},\mathbf{y}^{k},\mathbf{z}^{k},\boldsymbol\lambda^{k},\boldsymbol\theta^{k}) \label{ADMM_xt}\\
(\mathbf{y},\mathbf{z})^{k+1} &\! = &\! \underset{\mathbf{y},\mathbf{z}}{\mathrm{argmin}}~
 \mathcal{L}_{\rho}(\mathbf{x}^{k+1},\mathbf{t}^{k+1},\mathbf{y},\mathbf{z},\boldsymbol\lambda^{k},\boldsymbol\theta^{k}) \label{ADMM_yz}\\
 \boldsymbol\lambda^{k+1} &\! = &\! \boldsymbol\lambda^{k}+\rho(\mathbf{x}^{k+1}-\mathbf{y}^{k+1})  \label{ADMM_lambda} \\
 \boldsymbol\theta^{k+1} &\! = &\! \boldsymbol\theta^{k}+\rho(\mathbf{t}^{k+1}-\mathbf{z}^{k+1}).  \label{ADMM_theta}
\end{eqnarray}
According to the expression of the augmented Lagrangian function in (\ref{Lagrangian}), it follows from (\ref{ADMM_xt}) that the variables $\mathbf{x}$ and 
$\mathbf{t}$ are updated via solving the following non-convex problem
\begin{align}\label{ADMM_xt_upd}
\begin{aligned}
& \underset{\mathbf{x},\mathbf{t}}{\min} & &
 \|\mathbf{x}-\mathbf{y}^{k}+ \frac{\boldsymbol\lambda^{k}}{\rho}\|^{2} 
+\|\mathbf{t}-\mathbf{z}^{k}+\frac{\boldsymbol\theta^{k}}{\rho}\|^{2} \\
& \mathrm{s.t.}
& & (x_{i},t_{i}) \in \mathcal{X}, \quad i \in [n].
\end{aligned}
\end{align}
Exploiting the separable structure of (\ref{ADMM_xt_upd}), one immediately concludes that \eqref{ADMM_xt_upd} can be split into $n$ independent 2-dimensional problems that can be solved in parallel, i.e., for each $i \in [n]$,
\begin{equation}\label{xt_noncvx_proj}
(x_{i},t_{i})^{k+1} = \Pi_{\mathcal{X}}\left(y_{i}^{k}-\frac{\lambda_{i}^{k}}{\rho},z_{i}^{k}-\frac{\theta_{i}^{k}}{\rho}\right),
\end{equation}
where $\Pi_{\mathcal{X}}(.)$ denotes the Euclidean projection operator onto the set $\mathcal{X}$. 
Furthermore, (\ref{Lagrangian}) and (\ref{ADMM_yz}) imply that $\mathbf{y}$ and $\mathbf{z}$ are independently updated as follows:
\begin{eqnarray}
\mathbf{y}^{k+1} &=& \Pi_{\mathcal{Y}}\left( \mathbf{x}^{k+1} + \frac{\boldsymbol \lambda^{k}}{\rho} \right)  \label{y:upd}\\
\mathbf{z}^{k+1}  & = & \mathbf{t}^{k+1} +\frac{\boldsymbol\theta^{k}- \mathbf{1}}{\rho}.
\end{eqnarray}

\begin{algorithm}[t]
\caption{ADMM ($\rho > 0$)}
\label{alg:1}
{\small
\renewcommand{\thealgorithm}{}
\floatname{algorithm}{}
\begin{algorithmic}[1]
  \State Initialize: $\mathbf{y}^{0}$, $\mathbf{z}^{0}$, $\boldsymbol\lambda^{0}$, $\boldsymbol\theta^{0}$\;
   \For{$k \geq 0$}
        \State $(x_{i},t_{i})^{k+1} \gets \Pi_{\mathcal{X}}\left(y_{i}^{k}\!-\!\frac{\lambda_{i}^{k}}{\rho},z_{i}^{k}\!-\!\frac{\theta_{i}^{k}}{\rho}\right), \forall i \in [n]$\;
        \State $\mathbf{y}^{k+1} \gets \Pi_{\mathcal{Y}}\left( \mathbf{x}^{k+1} + \frac{\boldsymbol \lambda^{k}}{\rho} \right)$\; \label{step4_alg1}
        \State $\mathbf{z}^{k+1}  \gets  \mathbf{t}^{k+1} +\frac{\boldsymbol\theta^{k}- \mathbf{1}}{\rho}$\;
        \State $\boldsymbol\lambda^{k+1} \gets \boldsymbol\lambda^{k} + \rho(\mathbf{x}^{k+1} - \mathbf{y}^{k+1})$ \;
        \State $\boldsymbol\theta^{k+1} \gets \boldsymbol\theta^{k} + \rho(\mathbf{t}^{k+1} - \mathbf{z}^{k+1})$.
    \EndFor
\end{algorithmic}}%
\end{algorithm}
Algorithm \ref{alg:1} summarizes the proposed ADMM algorithm. It is clear that $\mathbf{z}$, $\boldsymbol\lambda$, and $\boldsymbol\theta$ merit closed-form updates. However, updating $(\mathbf{x},\mathbf{t})$ requires solving $n$ non-convex problems. 
Our strategy for dealing with this issue is presented in the section that follows.

\vspace{-0.3cm}
\subsection{Non-convex Projection} \label{non_cnvx_proj}
In this part, we present the method used to tackle the non-convex projection problem required to update $\mathbf{x}$ and $\mathbf{t}$.

Among the advantages of the proposed algorithm is that it is amenable to decentralization. As it is clear from (\ref{xt_noncvx_proj}), $\mathbf{x}$ and $\mathbf{t}$ can be updated element-wise via performing a projection operation onto the non-convex set $\mathcal{X}$, one for each $i \in [n]$. The $n$ projection problems can be run independently in parallel. We now outline the proposed idea for solving one such projection, i.e., we suppress the dependence on the index of the entry of $\mathbf{x}$ and $\mathbf{t}$. For $(\bar{x},\bar{t}) \in \mathbb{R}^{2}$, $\Pi_{\mathcal{X}}(\bar{x},\bar{t})$ entails solving
\begin{equation}\label{Pi_chi} 
\begin{aligned} 
& \underset{x,t}{\min}
& & g(x,t) \triangleq (t-\bar{t})^{2}+(x-\bar{x})^{2},
& \mathrm{s.t.}
& & t \geq |x|^{p}.
\end{aligned}
\end{equation}
If $\bar{t} \geq |\bar{x}|^{p}$, then trivially $\Pi_{\mathcal{X}}(\bar{x},\bar{t})=(\bar{x},\bar{t})$. Thus, we focus on the case in which 
$\bar{t} < |\bar{x}|^{p}$. The following proposition states the necessary optimality conditions for (\ref{Pi_chi}).
\begin{proposition}\label{Proposition1}
Let $\bar{t} < |\bar{x}|^{p}$, and $(x^{*},t^{*})$ be an optimal solution of (\ref{Pi_chi}). Then, the following properties are satisfied
\begin{enumerate}[(a)]
\begin{multicols}{2}
\item $\mathrm{sign}(x^{*})=\mathrm{sign}(\bar{x})$,
\item $t^{*} \geq \bar{t}$,
\item $|x^{*}|^{p} \geq \bar{t}$,
\item $t^{*}=|x^{*}|^{p}$.
\end{multicols}
\end{enumerate}
\end{proposition}
\begin{proof} We prove the statements by contradiction as follows:
\begin{enumerate}[(a)]
\item Suppose that $\mathrm{sign}(x^{*}) \neq \mathrm{sign}(\bar{x})$, then
\begin{align}
|x^{*}-\bar{x}|\!=\!|x^{*}-0|\!+\!|\bar{x}-0| > |\bar{x}-0|,
\end{align}
i.e., $(x^{*}\!-\!\bar{x})^{2} \!>\! (0\!-\!\bar{x})^{2}$. Hence, $g(x^{*},t^{*})-g(0,t^{*}) \!>\! 0$. Moreover, the feasibility of $(x^{*},t^{*})$ implies that $t^{*}>0$. Thus, $(0,t^{*})$ is feasible and attains a lower objective value than that attained by $(x^{*},t^{*})$. This contradicts the optimality of 
$(x^{*},t^{*})$.
\item Assume that $t^{*} < \bar{t}$. Then, 
\begin{equation}
g(x^{*},t^{*})-g(x^{*},\bar{t})=(t^{*}-\bar{t})^{2} > 0.
\end{equation}
Furthermore, by the feasibility of $(x^{*},t^{*})$, we have $|x^{*}|^{p} \leq t^{*} < \bar{t}$. Thus, $(x^{*},\bar{t})$ is feasible and attains a lower objective value than that attained by $(x^{*},t^{*})$. This contradicts the optimality of $(x^{*},t^{*})$.
\item Suppose that $|x^{*}|^{p} < \bar{t}$, i.e., 
\begin{equation}\label{|x^*|^p<t^bar}
-\bar{t}^{\frac{1}{p}} < x^{*}< \bar{t}^{\frac{1}{p}}.
\end{equation}
We now consider two cases, $\bar{x}>0$ and $\bar{x}<0$.
First, let $\bar{x} > 0$. Then, we have by $(\textit{a})$ and (\ref{|x^*|^p<t^bar}) that $0<x^{*}<\bar{t}^{\frac{1}{p}}$. Since $\bar{t}<|\bar{x}|^{p}$, i.e., 
$(\bar{x},\bar{t}) \notin \mathcal{X}$, therefore $\bar{t}^{\frac{1}{p}}<\bar{x}$ and hence, $0<x^{*}<\bar{t}^{\frac{1}{p}}<\bar{x}$. Pick $x_{0} > 0$ such that $|x_{0}|^{p}=\bar{t}$, i.e., $x_{0}=\bar{t}^{\frac{1}{p}}$. Then clearly, $x^{*} < x_{0} < \bar{x}$.
Thus, we have 
\begin{equation}\label{obj}
g(x^{*},t^{*})-g(x_{0},t^{*})=(x^{*}-\bar{x})^{2}-(x_{0}-\bar{x})^{2} > 0,
\end{equation}
where the last inequality follows the just proven identity that $x^{*} < x_{0} < \bar{x}$. Moreover, we have $|x_{0}|^{p}=\bar{t} \leq t^{*}$ by $(\textit{b})$. Thus, $(x_{0},t^{*})$ is feasible and attains a lower objective value than that attained by $(x^{*},t^{*})$. This contradicts the optimality of $(x^{*},t^{*})$. 
On the other hand, let $\bar{x}<0$. Then, we have by $(\textit{a})$ and (\ref{|x^*|^p<t^bar}) that 
$-\bar{t}^{\frac{1}{p}} < x^{*} < 0$. Since $\bar{t}<|\bar{x}|^{p}$, i.e., $(\bar{x},\bar{t}) \notin \mathcal{X}$, then $\bar{t}^{\frac{1}{p}}<|\bar{x}|$, i.e., $\bar{x}<-\bar{t}^{\frac{1}{p}}$. Therefore, $\bar{x} < -\bar{t}^{\frac{1}{p}} < x^{*}.$
Pick $x_{0}<0$ such that $|x_{0}|^{p}=\bar{t}$, i.e., $x_{0}=-\bar{t}^{\frac{1}{p}}$. Then, (\ref{obj}) also holds when 
$\bar{x}<0$. Note that $|x_{0}|^{p} = \bar{t} \leq t^{*}$ by $(\textit{b})$. Thus, $(x_{0},t^{*})$ is feasible and attains a lower objective value than that attained by $(x^{*},t^{*})$. This contradicts the optimality of $(x^{*},t^{*})$.
\item The feasibility of $(x^{*},t^{*})$ eliminates the possibility that $t^{*}<|x^{*}|^{p}$. Now let $t^{*}>|x^{*}|^{p}$ and pick $t_{0}=|x^{*}|^{p}$. Then, $\bar{t} \leq |x^{*}|^{p}=t_{0}<t^{*}$, where the first inequality follows from 
$(\textit{c})$. Then, $0 \leq t_{0}-\bar{t} < t^{*}-\bar{t}$. 
Thus, we have
\begin{equation}
g(x^{*},t^{*})-g(x^{*},t_{0})=(t^{*}-\bar{t})^{2}-(t_{0}-\bar{t})^{2} > 0, 
\end{equation}
Furthermore, the feasibility of $(x^{*},t_{0})$ follows trivially from the choice of $t_{0}$. Thus, 
$(x^{*},t_{0})$ is feasible and attains a lower objective value than that attained by $(x^{*},t^{*})$. This contradicts the optimality of $(x^{*},t^{*})$. 
\end{enumerate}
This concludes the proof.
\end{proof}
\vspace{-0.2cm}
We now make use of the fact that for (\ref{Pi_chi}), an optimal solution $(x^{*},t^{*})$ satisfies $t^{*}=|x^{*}|^{p}$ and hence, (\ref{Pi_chi}) reduces to solving
\begin{equation}\label{Pi_chi_scalar}  
\underset{x}{\min}
\quad (|x|^{p}-\bar{t})^{2}+(x-\bar{x})^{2}.
\end{equation}
The first order necessary optimality condition for (\ref{Pi_chi_scalar}) implies the following:
\begin{equation}\label{first_order}
p|x^{*}|^{p-1}\mathrm{sign}(x^{*})(|x^{*}|^{p}-\bar{t})+x^{*}-\bar{x}=0.
\end{equation}
By the symmetry of the function $|x|^{p}$, without loss of generality, assume that $x^{*}>0$ and let $0 < p=\frac{s}{q} < 1$ for some $s,q \in \mathbb{Z}_{+}$. A change of variables $a^{q}=x^{*}$ plugged in (\ref{first_order}) shows that finding an optimal solution for (\ref{Pi_chi}) reduces to finding a root of the following scalar degree $2q$ polynomial:
\begin{equation}\label{polynomial}
a^{2q}+\frac{s}{q}\left(a^{2s}-\bar{t}a^{s}\right)-\bar{x}a^{q}.
\end{equation}
Thus, to find $\Pi_{\mathcal{X}}(\bar{x},\bar{t})$, solve for a root $a^{*}$ of the polynomial in (\ref{polynomial}) such that $(a^{*^q},a^{*^{s}})$ minimizes $g(x,t)$. Algorithm \ref{alg:2} summarizes the method we use to solve problem \eqref{Pi_chi}. In case $\bar{x}=0$, we set $x^{*}=t^{*}=0$. If the set 
$\bar{\mathcal{R}}$ is empty, we set $x^{*}=0$ and $t^{*}=(\bar{t})^{+}$.
 
 \begin{algorithm}[t]
\caption{Non-convex projection ($p=\frac{s}{q}<1$)}
\label{alg:2}
\renewcommand{\thealgorithm}{}
\floatname{algorithm}{}
\begin{algorithmic}[1]
    \State $\mathcal{R} \gets \mathrm{roots} \{ a^{2q}+\frac{s}{q}(a^{2s}-\bar{t}a^{s})-|\bar{x}|a^{q} \}$ 
    \State $\bar{\mathcal{R}} \gets \mathcal{R} \setminus \{\mathrm{complex~numbers~and~negative~reals~in~}\mathcal{R}\}$ 
    \State $\mathcal{T} \gets \{(r^{q},r^{s}) : r \in \bar{\mathcal{R}}\}$ 
    \State $(\hat{x},t^{*}) \gets \mathrm{argmin} ~ \{ g(x,t) : (x,t) \in \mathcal{T} \}$ 
    \State $x^{*} \gets \mathrm{sign}(\bar{x}) \hat{x}$
\end{algorithmic}\end{algorithm}

\subsection{Convex Projection} \label{vec cnvx projec}
The convex projection for $\textbf{y}$-update in (\ref{y:upd}) can be formulated as the following convex optimization problem
\begin{equation} \label{conv_projec_prob}  
\begin{aligned} 
\textbf{y}^{k+1}=\argmin_{\textbf{y}} \quad & \left\lVert \textbf{y}-(\mathbf{x}^{k+1}+\frac{\boldsymbol\lambda^{k}}{\rho}) \right\rVert^{2}, \quad
\textrm{s.t.}  & \textbf{y} \in \mathcal{V},\\
\end{aligned}
\end{equation}
where $\left\lVert.\right\rVert$ is the euclidean norm. Convex problems can be solved by a variety of contemporary methods including bundle methods \cite{18}, sub-gradient projection \cite{19}, interior point methods \cite{20}, and ellipsoid methods \cite{21}. The efficiency of optimization techniques rely mainly on exploiting the structure of the constraint set. As mentioned in \ref{contributions}, to be general, we aim to solve the problem in \eqref{motivProb} with no assumptions on the set $\mathcal{V}$, other than it being closed and convex. That said, 
if possible, through exploiting the structure of $\mathcal{V}$, one should be able to reduce the computational complexity of solving  \eqref{conv_projec_prob}.
\begin{remark}
As per our knowledge, none of the existing literature considered the convergence of an ADMM algorithm for solving the general problem in \eqref{motivProb}. As discussed in \ref{SVR}, on one hand, the work in \cite{ref4} studied the convergence of ADMM under mild assumptions. However, assuming $\mathcal{V}$ has a particular form, these assumptions hold only 
if the function $f$ 
defining the the constraint set $\mathcal{V}=\{\mathbf{x}:\ f(\mathbf{x})\leq 0\}$ in \eqref{motivProb} is Lipschitz differentiable. On the other hand, \cite{ref5} studied the convergence of a non ADMM algorithm to solve \eqref{motivProb} while assuming that the global optimal for each update step can be found efficiently.
\end{remark}
\vspace{-0.4cm}
\section{Rank Minimization Algorithm} \label{section 2}
We consider the same problem as in (\ref{1}) and propose a method for approximating its solution efficiently. The Schatten-p heuristic of (\ref{1}) can be written as 
\begin{equation} \label{2}
\begin{aligned}
\min_{\mathbf{X}} \quad & \left\lVert \mathbf{X}\right\rVert_{p}^{p}\stackrel{\Delta}{=}\sum_{i=1}^{\mathrm{L}}\lvert\sigma_{i}(\mathbf{X})\rvert^{p}, \quad
\textrm{s.t.}  & \mathbf{X} \in \mathcal{M},\\
\end{aligned}
\end{equation}
where $\mathrm{L}=\min(m,n)$ and $\sigma_{i}(\mathbf{X})$ is the $i$th singular value of $\mathbf{X}$. When $p=1$, problem (\ref{2}) is a convex one which is eventually the nuclear norm heuristic. We consider a non-convex case where $0<p<1$, which has the corresponding epi-graph form,
\begin{equation} \label{3}
\begin{aligned}  
\min_{\mathbf{X},\mathbf{t}} \quad & \mathbf{1}^{\top}\mathbf{t}, \\
\textrm{s.t.} \quad & \lvert \sigma_{i}(\mathbf{X}) \rvert^{p}\leq t_{i}, \quad i \in \{1,\dots \mathrm{L} \}, \quad \mathbf{X} \in \mathcal{M},
\end{aligned}
\end{equation}
such that $\mathbf{t}=[t_{i}]_{i\in [\mathrm{L}]}$. Defining the epi-graph set $\mathring{\mathcal{X}}$ for the function $\sigma(X)$, where $\mathring{\mathcal{X}}\stackrel{\Delta}{=}\{(\sigma(\mathbf{X}),t)\!\in\! \mathbb{R}^{2}\!:\!\lvert\sigma(\mathbf{X})\rvert^{p}\!\leq\! t\} \subseteq \mathbb{R}^{2}$, the problem in (\ref{3}) can be written as,
\begin{equation} \label{4}
\begin{aligned}
\min_{\mathbf{X},\mathbf{t}} \quad & \mathbf{1}^{\top}\mathbf{t}+\mathbbm{1}_{\mathcal{M}}(\mathbf{X})+\sum_{i=1}^{\mathrm{L}}\mathbbm{1}_{\mathring{\mathcal{X}}}(\sigma_{i}(\mathbf{X}), t_{i}).
\end{aligned}
\end{equation}

In order to structure the problem in a from that ADMM can exploit, we introduce the auxiliary variables $\mathbf{Y} \in \mathbb{R}^{m\times n}$ and $\mathbf{z}=[z_{i}]_{i\in [\mathrm{L}]}$ which makes the problem in (\ref{4}) be,
\begin{equation} \label{5}
\begin{aligned}
\min_{\mathbf{X},\mathbf{t},\mathbf{Y},\mathbf{z}} \quad & \mathbf{1}^{\top}\mathbf{z}+\mathbbm{1}_{\mathcal{V}}(\mathbf{Y})+\sum_{i=1}^{\mathrm{L}}\mathbbm{1}_{\mathring{\mathcal{X}}}(\sigma_{i}(\mathbf{X}), t_{i}), \\
\textrm{s.t.} \quad & \mathbf{X}=\mathbf{Y}: \boldsymbol\Lambda, \quad
\mathbf{t}=\mathbf{z}: \boldsymbol\theta, 
\end{aligned}
\end{equation}
such that $\boldsymbol\Lambda$, $\boldsymbol\theta$ are the dual variables associated with $\mathbf{X}$ and $\mathbf{t}$ respectively. Similar to (\ref{Lagrangian}), the Lagrangian function associated with (\ref{5}) augmented with a quadratic penalty for the equality constraint violation with a parameter $\rho>0$, is  
\begin{equation}
\begin{aligned}
    &\mathcal{L}_{\rho}(\mathbf{X},\mathbf{Y},\mathbf{t},\mathbf{z},\boldsymbol\Lambda,\boldsymbol\theta)\!=\!\mathbf{1}^{\top}\mathbf{z}\!+\!\mathbbm{1}_{\mathcal{M}}(\mathbf{Y})\!+\!\sum_{i=1}^{\mathrm{L}}\mathbbm{1}_{\mathring{\mathcal{X}}}(\sigma_{i}(\mathbf{X}), t_{i})\\&+\!Tr\{\boldsymbol\Lambda^{\top} (\mathbf{X}\!-\!\mathbf{Y})\}\!+\!\boldsymbol\theta^{\top}(\mathbf{t}\!-\!\mathbf{z})\!+\!\frac{\rho}{2}(\left\lVert \mathbf{X}\!-\!\mathbf{Y} \right\rVert_{\mathrm{f}}^{2}\!+\!\left\lVert \mathbf{t}\!-\!\mathbf{z}\right\rVert^{2}),
\end{aligned}
\end{equation}
where $Tr\{.\}$ is the trace operator. Considering the 2-tuples $(\mathbf{X},\mathbf{t})$ and $(\mathbf{Y},\mathbf{z})$, the ADMM iterations is,
\begin{align}
    &(\mathbf{X},\mathbf{t})^{k+1}=\argmin_{\mathbf{X},\mathbf{t}} \mathcal{L}_{\rho}(\mathbf{X},\mathbf{Y}^{k},\mathbf{t},\mathbf{z}^{k},\boldsymbol\Lambda^{k},\boldsymbol\theta^{k}), \label{6} \\ 
    &\mathbf{Y}^{k+1}=\argmin_{\mathbf{Y}} \mathcal{L}_{\rho}(\mathbf{X}^{k+1},\mathbf{Y},\mathbf{t}^{k+1},\mathbf{z}^{k},\boldsymbol\Lambda^{k},\boldsymbol\theta^{k}), \label{12}\\
    &\mathbf{z}^{k+1}=\argmin_{\mathbf{z}} \mathcal{L}_{\rho}(\mathbf{X}^{k+1},\mathbf{Y}^{k+1},\mathbf{t}^{k+1},\mathbf{z},\boldsymbol\Lambda^{k},\boldsymbol\theta^{k}), \label{14}\\
    &\boldsymbol\Lambda^{k+1}=\boldsymbol\Lambda^{k}+\rho(\mathbf{X}^{k+1}-\mathbf{Y}^{k+1}), \\
    &\boldsymbol\theta^{k+1}=\boldsymbol\theta^{k}+\rho(\mathbf{t}^{k+1}-\mathbf{z}^{k+1}).
\end{align}
\subsection{$(\mathbf{X},\mathbf{t})$ update}
By completing the square and with some simple algebra, it can be shown that the problem in (\ref{6}) is equivalent to
\begin{equation}  \label{7}
\begin{aligned}
\min_{\mathbf{X},\mathbf{t}} \quad & \left\lVert \mathbf{X}-\Bar{\mathbf{X}}^{k} \right\rVert_{\mathrm{f}}^{2}+\left\lVert \mathbf{t}-\mathbf{\Bar{t}}^{k}\right\rVert^{2}, \\
\textrm{s.t.} \quad & \lvert \sigma_{i}(\mathbf{X}) \rvert^{p}\leq t_{i}, \quad i \in \{1,\dots \mathrm{L} \},\\
\end{aligned}
\end{equation}
where $\Bar{\mathbf{X}}^{k}\stackrel{\Delta}{=}\mathbf{Y}^{k}-\frac{\boldsymbol\Lambda^{k}}{\rho}$ and $\mathbf{\Bar{t}}^{k}\stackrel{\Delta}{=}\mathbf{z}^{k}-\frac{\boldsymbol\theta^{k}}{\rho}$. For an ease of notations, we will drop the iteration index $k$. Assume that $\mathbf{X}=\mathbf{P}\mathbf{\Sigma} \mathbf{Q}^{\top}$ and $\Bar{\mathbf{X}}=\mathbf{U}\mathbf{\Delta} \mathbf{V}^{\top}$ is the singular value decomposition (SVD) of $\mathbf{X}$ and $\Bar{\mathbf{X}}$ respectively. Where $\mathbf{\Sigma}, \mathbf{\Delta} \in \mathbb{R}^{\mathrm{L}\times \mathrm{L}}$ are diagonal matrices with the singular values associated $\mathbf{X}$ and $\Bar{\mathbf{X}}$ while $\mathbf{P}, \mathbf{U} \in \mathbb{R}^{m\times \mathrm{L}}$ and $\mathbf{Q}, \mathbf{V} \in \mathbb{R}^{n\times \mathrm{L}}$ are the unitary matrices. By applying the same steps as in Theorem 3 of \cite{12}, we can write the first term of (\ref{7}) after dropping $k$ as,
\begin{equation}
    \begin{aligned}
    &\left\lVert \mathbf{X}-\Bar{\mathbf{X}} \right\rVert_{\mathrm{f}}^{2}=\left\lVert \mathbf{P}\mathbf{\Sigma} \mathbf{Q}^{\top}-\mathbf{U}\mathbf{\Delta} \mathbf{V}^{\top} \right\rVert_{\mathrm{f}}^{2} \\
    &=\left\lVert \mathbf{P}\mathbf{\Sigma} \mathbf{Q}^{\top} \right\rVert_{\mathrm{f}}^{2}+\left\lVert \mathbf{U}\mathbf{\Delta} \mathbf{V}^{\top} \right\rVert_{\mathrm{f}}^{2}-2T\{\mathbf{X}^{\top}\Bar{\mathbf{X}}\}\\
    &\stackrel{(\mathrm{a})}{=}\!Tr\{ \mathbf{\Sigma}^{\top}\!\mathbf{\Sigma} \}\!+\!Tr\{ \mathbf{\Delta}^{\top}\!\mathbf{\Delta} \}\!-\!2Tr\{\mathbf{Q}\mathbf{\Sigma}^{\top}\!\mathbf{P}^{\top}\!\mathbf{U}\!\mathbf{\Delta}\!\mathbf{V}^{\top}\!\}\\
    &\stackrel{(\mathrm{b})}{\geq}Tr\{ \mathbf{\Sigma}^{\top}\! \mathbf{\Sigma} \}\!+\!Tr\{ \mathbf{\Delta}^{\top}\! \mathbf{\Delta} \}\!-\!2Tr\{ \mathbf{\Sigma}^{\top}\!\mathbf{\Delta} \}
    \!=\!\left\lVert \mathbf{\Sigma}\!-\!\mathbf{\Delta} \right\rVert_{\mathrm{f}}^{2},
    \end{aligned}
\end{equation}
where (a) is because $\mathbf{P}^{\top}\mathbf{P}=\mathbf{Q}^{\top}\mathbf{Q}=\mathbf{U}^{\top}\mathbf{U}=\mathbf{V}^{\top}\mathbf{V}=\mathbf{I}_{\mathrm{L}\times \mathrm{L}}$ with $\mathbf{I}_{\mathrm{L}\times \mathrm{L}}$ being an identity matrix of size $\mathrm{L}$, and exploiting the circular property of the trace while (b) holds is from the main result of \cite{13}. In order to make $\left\lVert \mathbf{X}-\Bar{\mathbf{X}}^{k} \right\rVert_{\mathrm{f}}^{2}$ achieve its derived lower bound, we set $\mathbf{P}=\mathbf{U}$ and $\mathbf{Q}=\mathbf{V}$. 

Henceforth, the problem in (\ref{7}) will be equivalent to,
\begin{equation}  \label{8}
\begin{aligned}
\min_{\mathbf{X},\mathbf{t}} \quad & \left\lVert \mathbf{x}-\mathbf{\Bar{x}} \right\rVert^{2}+\left\lVert \mathbf{t}-\mathbf{\Bar{t}}\right\rVert^{2}, \\
\textrm{s.t.} \quad & \lvert x_{i} \rvert^{p}\leq t_{i}, \quad i \in \{1,\dots \mathrm{L} \},\\
\end{aligned}
\end{equation}
where $\mathbf{x}=[x_{i}]_{i\in [\mathrm{L}]}$ and $\mathbf{\Bar{x}}=[\bar{x}_{i}]_{i\in [\mathrm{L}]}$ are the vectors of singular values of the matrices $\mathbf{X}$ and $\Bar{\mathbf{X}}$ respectively. The optimal solution $\mathbf{X}^{*}$ for (\ref{7}) can be calculated by finding the optimal $\mathbf{x}^{*}$ of (\ref{8}) and then $\mathbf{X}^{*}=\mathbf{U} \mathbf{\Sigma}^{*} \mathbf{V}^{T}$, where $\mathbf{\Sigma}^{*}=$diag$(\mathbf{x}^{*})$ and diag$(.)$ is an operator that converts a vector to its corresponding diagonal matrix. Since the problem in (\ref{8}) is separable, we drop the index $i$ and only consider solving
\begin{equation}  \label{9}
\begin{aligned}
\min_{x,t} \quad & (x-\Bar{x})^{2}+(t-\Bar{t})^{2}, \quad
\textrm{s.t.}  & \lvert x \rvert^{p}\leq t.\\
\end{aligned}
\end{equation}

It can be realized that (\ref{9}) is the same as (\ref{Pi_chi}), hence, its optimal solution can be found by applying algorithm \ref{alg:2}.
\subsection{$(\mathbf{Y},\mathbf{z})$ update}
After updating $(\mathbf{X},\mathbf{t})$ while fixing $\boldsymbol\Lambda$ and $\boldsymbol\theta$, the problem in (\ref{12}) can be written as,
\begin{equation}  \label{13}
\begin{aligned}
\mathbf{Y}^{k+1}\!=\!\argmin_{\mathbf{Y}} & \left\lVert \mathbf{Y}\!-\!(\mathbf{X}^{k+1}\!+\!\frac{\boldsymbol\Lambda^{k}}{\rho}) \right\rVert_{\mathrm{f}}^{2}, 
\textrm{s.t.} \quad \mathbf{Y} \in \mathcal{M},\\
\end{aligned}
\end{equation}
which is clearly a convex optimization problem representing the projection of the point $\mathbf{X}^{k+1}+\frac{\boldsymbol\Lambda^{k}}{\rho}$ on the set $\mathcal{M}$ and can be solved by various known class of algorithms as discussed in section \ref{vec cnvx projec}. 

Upon updating $\mathbf{Y}$, the $\mathbf{z}$ update in (\ref{14}) is
\begin{equation}  
\begin{aligned}
\mathbf{z}^{k+1}=\argmin_{\mathbf{z}} \quad & \mathbf{1}^{\top}\mathbf{z}+\frac{\rho}{2}\left\lVert \mathbf{z}-(t^{k+1}+\frac{\boldsymbol\theta^{k}}{\rho}) \right\rVert^{2}, \\
\end{aligned}
\end{equation}
which has the closed-form solution $\mathbf{z}=\mathbf{t}^{k+1}+\frac{\boldsymbol\theta^{k}-1}{\rho}$.

\section{Proximal Gradient Algorithm}
The SVR algorithm 
deals with the $\ell_{p}$ relaxation of \eqref{zero norm} 
without assuming any specific structure for $\mathcal{V}$, other than being closed and convex. 
Indeed, the algorithm only requires the Euclidean projections onto $\mathcal{V}$ as in \eqref{conv_projec_prob}. However, this approach suffers from two pitfalls: 1) 
high computational complexity per iteration as a result of solving \eqref{conv_projec_prob} in every iteration, and 2) the lack of 
convergence guarantees. 

In this 
section, we consider a 
sub-class of problems 
with a specific structure for the convex set of the form $\mathcal{V}=\{\mathbf{x}:~f(\mathbf{x})\leq 0\}$, where 
$f(\mathbf{x})=\left\lVert A\mathbf{x}-b\right\rVert^{2}-\epsilon$ for some given $\epsilon\geq 0$, $A\in \mathbb{R}^{m\times n}$ and $b\in\mathbb{R}^{m}$. Note that 
$f(\mathbf{x})$ is a convex function with Lipschitz continuous gradient. i.e., $f$ is $L$-smooth: $\left\lVert\nabla f(\mathbf{x})-\nabla f(\mathbf{y})\right\rVert \leq L\left\lVert \mathbf{x}-\mathbf{y}\right\rVert$ for all $\mathbf{x},\mathbf{y} \in \mathbb{R}^{n}$ and $L\triangleq\|A\|^2$. Specifically, 
in order to solve
\begin{equation} \label{PG_intro}
\begin{aligned}
\min_{\mathbf{x}} \quad &  \left\lVert \mathbf{x}\right\rVert_{p}^{p}, \quad
\textrm{s.t.} & f(\mathbf{x})\leq 0,
\end{aligned}
\end{equation}
we aim to develop an 
efficient algorithm 
with some convergence guarantees 
for the following Lagrangian relaxation: 
\begin{equation} \label{PG}
\begin{aligned} 
& \underset{\mathbf{x}}{\min}
& F(\mathbf{x})\stackrel{\Delta}{=}\left\lVert \mathbf{x}\right\rVert_{p}^{p}+\frac{\mu}{2}f(\mathbf{x}), 
\end{aligned}
\end{equation}
where $\mu\geq 0$ is 
the dual multiplier that captures the trade-off between solution sparsity and fidelity.

A canonical problem for the regularized risk minimization has the following form:
\begin{equation} \label{orig_PG}
\begin{aligned} 
& \underset{\mathbf{x}}{\min}
& g(\mathbf{x})+h(\mathbf{x})
\end{aligned}
\end{equation}
where $h$ is an $L$-smooth loss function and $g$ is a a regularizer term. When both $g$ and $h$ are convex, the proximal gradient (PG) algorithm \cite{PG1} 
can compute a solution to~\eqref{orig_PG} through iteratively taking PG steps, i.e., $\mathbf{x}^{k+1}=\mathbf{prox}_{g/\lambda}(\mathbf{x}^{k}-\nabla h(\mathbf{x}^{k})/L)$ where $\mathbf{prox}_{g/\lambda}(.)\stackrel{\Delta}{=}\argmin_{\mathbf{x}} g(\mathbf{x}) + \frac{\lambda}{2} \left\lVert \mathbf{x}-\cdot\right\rVert^{2}$, for some constant $\lambda$. 
When $g$ is convex, prox operation is well-defined; thus, the PG step can be computed. 

Comparing both \eqref{PG} and \eqref{orig_PG}, 
the convexity assumption of $g(\mathbf{x})$ in \eqref{orig_PG} is not satisfied 
for $\|\mathbf{x}\|_{p}^{p}$ in \eqref{PG}. 
When the regularizer is a continuous nonconvex function, the proximal map $\mathbf{prox}_{g/\lambda}$ may not exist, let alone it can be computed in closed form.
On the other hand, for $\|\mathbf{x}\|_{p}^{p}$, using similar arguments for the non-convex projection step introduced in subsection~\ref{non_cnvx_proj}, we aim to derive an analytical solution that can be computed efficiently.
Indeed, assuming $p\in(0,1)$ is a positive rational number, the proposed method for computing the proximal map of $\|\mathbf{x}\|_{p}^{p}$ involves finding the roots of a polynomial of order $2q$, where $q\in\mathbb{Z}_+$ 
such that $p=s/q$ for some $s\in\mathbb{Z}_+$.

Since $f$ is $L$-smooth, for all $\mathbf{x}, \mathbf{y} \in \mathbb{R}^{n}$, we have 
\begin{equation} \label{Lipshtiz_ineq}
    f(\mathbf{x})\leq f(\mathbf{y})+\nabla f(\mathbf{y})^{\top}(\mathbf{x}-\mathbf{y})+\frac{L}{2}\left\lVert \mathbf{x}-\mathbf{y}\right\rVert^{2}.
\end{equation}
Given $\mathbf{x}^k$, replacing $f(\mathbf{x})$ with the upper bound in~\eqref{Lipshtiz_ineq} for $\mathbf{y}=\mathbf{x}^k$, the prox-gradient operation naturally arises as follows:
\begin{equation} \label{PG_update}  
\begin{aligned} 
\mathbf{x}^{k+1}=\argmin_{\mathbf{X}}\left\lVert \mathbf{x}\right\rVert_{p}^{p}+\frac{\mu}{2}&[f(\mathbf{x}^{k})+\nabla f(\mathbf{x}^{k})^{\top}(\mathbf{x}-\mathbf{x}^{k}) \\ &+\frac{L}{2}\left\lVert \mathbf{x}-\mathbf{x}^{k}\right\rVert^{2}].
\end{aligned}
\end{equation}
By completing the square, \eqref{PG_update} yields to
\begin{equation} \label{PG_update_completed}  
\begin{aligned} 
\mathbf{x}^{k+1}=\argmin_{\mathbf{X}}\left\lVert \mathbf{x}\right\rVert_{p}^{p}+\frac{\mu L}{4}\left\lVert \mathbf{x}-\left(\mathbf{x}^{k}-\frac{1}{L}\nabla f(\mathbf{x^{k}})\right)\right\rVert^{2}.
\end{aligned}
\end{equation}
Defining $\bar{\mathbf{x}}^{k}\stackrel{\Delta}{=}\mathbf{x}^{k}-\frac{1}{L}\nabla f(\mathbf{x^{k}})$, \eqref{PG_update_completed} can be rewritten as
\begin{equation}
\begin{aligned} 
\mathbf{x}^{k+1}&=\argmin_{\mathbf{X}}\left\lVert \mathbf{x}\right\rVert_{p}^{p}+\frac{\mu L}{4}\left\lVert \mathbf{x}-\bar{\mathbf{x}}^{k}\right\rVert^{2} \\
&=\argmin_{\mathbf{X}}\sum_{i=1}^{n}|x_{i}|^{p}+\frac{\mu L}{4}(x_{i}-\bar{x}_{i}^{k})^{2},
\end{aligned}
\end{equation}
which is clearly a separable structure in the entries of $\mathbf{x}$. Therefore, for each 
$i\in [n]$, we have
\begin{equation} \label{sep_obj}
\begin{aligned} 
x_{i}^{k+1}\!=\!\argmin_{x_{i}}|x_{i}|^{p}\!+\!\frac{\mu L}{4}(x_{i}\!-\!\bar{x}_{i}^{k})^{2} 
=\mathbf{prox}_{\bar{g}/\frac{\mu L}{2}}(\bar{x}_{i}^{k}),
\end{aligned}
\end{equation}
where $\bar{g}:\mathbb{R}\to\mathbb{R}_+$ such that $\bar{g}(t)=|t|^{p}$ for some positive rational $p\in(0,1)$. 

Next, we consider a generic form of \eqref{sep_obj}, i.e., given some $\bar{t}\in\mathbb{R}$, we would like to compute \begin{equation}
\label{eq:generic-prox}
    t^*=\argmin_{t}\{|t|^p+\frac{\mu L}{4}(t-\bar{t})^2\}.
\end{equation}
The first-order optimality condition for \eqref{eq:generic-prox} can be written as
\begin{equation} \label{first_order_optimiality_condition}
\begin{aligned} 
p|t^{*}|^{p-1}\text{sign}(t^*)+\frac{\mu L}{2}(t^{*}-\bar{t})=0. 
\end{aligned}
\end{equation}
Using similar arguments with those in section \ref{non_cnvx_proj} for Proposition~\ref{Proposition1}, we can conclude that the optimal solution $t^{*}$ attains the property that $\text{sign}(t^{*})=\text{sign}(\bar{t})$. 
Without loss of generality, exploiting the symmetry of the function $\bar{g}$, we only consider the case when $\bar{t}>0$; hence, the optimal solution $t^{*}$ is the smallest positive root of the following polynomial:
\begin{equation} \label{first_order_optimiality_condition_reduced}
\begin{aligned} 
p|t^{*}|^{p-1}+\frac{\mu L}{2}(t^{*}-\bar{t})=0.
\end{aligned}
\end{equation}
As in \eqref{polynomial}, suppose $0 < p=\frac{s}{q} < 1$ for some $s,q \in \mathbb{Z}_{+}$. Using the change of variables $a\triangleq (t^{*})^{\tfrac{1}{q}}$, 
\eqref{first_order_optimiality_condition_reduced} reduces to finding the roots of a polynomial of degree $2q$: 
\begin{equation} \label{roots_prob_for_PG}
\begin{aligned} 
a^{2q}-\bar{t}a^{q}+\frac{2s}{q\mu L}a^{s}=0.
\end{aligned}
\end{equation} 

\begin{algorithm}[t]
\caption{Accelerated PG algorithm}
\label{PG_alg}
{\small
\renewcommand{\thealgorithm}{}
\floatname{algorithm}{}
\begin{algorithmic}[1]
    \State Initialize: $\mu$, $s=1$, $q=2$, $l$, $\mathbf{x}^{0}$, $\mathbf{x}^{1}$, $k=1$.
    \Repeat
    \State $\mathbf{y}^{k}=\mathbf{x}^{k}+\frac{k-1}{k+2}(\mathbf{x}^{k}-\mathbf{x}^{k-1})$
    \State $\Delta^{k}=\max_{t=\max\{1,k-l\},\dots, k} F(\mathbf{x}^{t})$
    \If{$F(\mathbf{y}^{k})\leq \Delta^{k}$}:
        \State $\mathbf{v}^{k}=\mathbf{y}^{k}$
    \Else:
        \State $\mathbf{v}^{k}=\mathbf{x}^{k}$
    \EndIf
    \State $\bar{\mathbf{x}}^{k}=\mathbf{v}^{k}-\frac{1}{L}\nabla f(\mathbf{v}^{k})$
    \For{\texttt{$i \in [n]$}}:
        \State \texttt{solve} $a^{2q}-\bar{x}_{i}a^{q}+\frac{2s}{q \mu L}a^{s}=0$
        \State $x_{i}^{k+1}=a^{*^{q}}$
    \EndFor
    \State $k=k+1$
    \Until{convergence}
\end{algorithmic}}%
\end{algorithm}
To efficiently solve \eqref{PG_intro}, we will use Algorithm \ref{PG_alg}, which is an implementation of nonconvex inexact accelerated proximal gradient (APG) descent method proposed in \cite[Algorithm~2]{APG_method}. 
To summarize, \cite[Algorithm~2]{APG_method} is designed to solve composite problems of the form in~\eqref{orig_PG} assuming that $h$ is $L$-smooth and $g$ is proper lower-semicontinuous such that $F\triangleq h+g$ is bounded from below and coercive, i.e., $\lim_{\Vert \Vert\to\infty}F()=+\infty$ -- note that there is no assumption regarding neither $h$ nor $g$ to be convex. The key points enhancing both practical behavior of and theoretical guarantees for \cite[Algorithm~2]{APG_method} can be summarized as given below: \looseness=-1
\begin{itemize}
    \item An extrapolation $\mathbf{y}_{k}$ is generated as introduced in \cite{APG} for the APG algorithm (step 3).
    \item Steps 4 through 9 allow non monotone update of the objective. $F(\mathbf{y}_{k})$ is checked with respect to the maximum of the latest $l$ objective values. The gradient step is adjusted according to this (step 9). This permits $\mathbf{y}^{k}$ to occasionally increase the objective and makes $F(\mathbf{y}^{k})$ be less than the maximum of the objective value of the latest $l$ iterations. 
    \item Steps 11 and 12 are the solution of the PG step using the non-convex projection method. 
\end{itemize}
In the next part, we show that algorithm \ref{PG_alg} converges to a critical point and it exhibits a convergence rate of $O(\frac{1}{k})$, where $k$ is the iteration budget that is given to the algorithm. 
\begin{definition} (\cite{attouch2013convergence})
The Frechet sub-differential of $F$ at $\mathbf{x}$ is
\begin{equation}
    \hat{\partial}F(\mathbf{x})\stackrel{\Delta}{=}\Big\{\mathbf{u}: \lim_{\mathbf{y}\neq\mathbf{x}}\lim_{\mathbf{y}\to\mathbf{x}} \frac{F(\mathbf{y})-F(\mathbf{x})-\mathbf{u}^{\top}(\mathbf{y}-\mathbf{x})}{\left\lVert\mathbf{y}-\mathbf{x}\right\rVert}\geq0\Big\}.
\end{equation}
The sub-differential of $F$ at $\mathbf{x}$ is 
\begin{equation}
\begin{aligned}
    \partial F(\mathbf{x})\stackrel{\Delta}{=}\{&\mathbf{u}: \exists \mathbf{x}^{k}\to \mathbf{x}, F(\mathbf{x}^{k})\to F(\mathbf{x})\enspace\text{and}\enspace\mathbf{u}^{k}\in \hat{\partial}F(\mathbf{x}^{k}) \\ 
    &\to \mathbf{u}\enspace\text{as}\enspace k\to \infty \}.
\end{aligned}
\end{equation}
\end{definition}
\begin{definition} \label{def2}
(\cite{attouch2013convergence}) $\mathbf{x}$ is a critical point of $F$ if $0\in \partial g(\mathbf{x})+\nabla h(\mathbf{x})$. 
\end{definition}
By comparing \eqref{orig_PG} and \eqref{PG}, it can be realized that the functions $g(\mathbf{x})$ and $h(\mathbf{x})$ in definition \ref{def2} are equal to $\left\lVert \mathbf{x}\right\rVert_{p}^{p}$ and $\frac{\mu}{2}f(\mathbf{x})$ respectively.
\begin{theorem}
The sequence $\mathbf{x}^{k}$ generated from algorithm \ref{PG_alg} has at least one limit point and all the generated limit points are critical points of \eqref{PG}. Moreover, the algorithm converges with rate $O(\frac{1}{K})$, where $K$ is the iteration budget given to the algorithm. \looseness=-1
\end{theorem}
\begin{proof}
It can easily be verified that our problem in \eqref{PG} satisfies all required assumptions for Algorithm~\ref{PG_alg}.
Indeed, 
\begin{enumerate}
    \item The function $g(\mathbf{x})=\left\lVert \mathbf{x}\right\rVert_{p}^{p}$ is a proper and lower semi-continuous function.
    \item The gradient of $h(\mathbf{x})=\frac{\mu}{2}f(\mathbf{x})$ is $\bar{L}$-Lipschitz smooth, i.e., $\left\lVert\nabla h(\mathbf{x})-\nabla h(\mathbf{y})\right\rVert \leq \bar{L}\left\lVert \mathbf{x}-\mathbf{y}\right\rVert$ for all $\mathbf{x},\mathbf{y} \in \mathbb{R}^{n}$, with $\bar{L}=\frac{\mu}{2}L$.
    \item $F(\mathbf{x})\!=\!g(\mathbf{x})\!+\!h(\mathbf{x})$ is bounded from below, i.e., $F(\mathbf{x})\!\geq\!0$.
    \item $\lim_{\left\lVert\mathbf{x}\right\rVert\to\infty}F(\mathbf{x})=\infty$. 
    \item The introduced non-convex projection method is an exact solution for the proximal gradient step. This is because it is based on finding the roots of a polynomial of order $2q$ in equation \eqref{roots_prob_for_PG}. 
\end{enumerate}
Therefore, the assumptions required for theorem 4.1 for critical point convergence and proposition 4.3 for the rate of convergence in \cite{APG_method} are satisfied which then completes the proof.  
\end{proof}
\vspace{-0.3cm}
\begin{remark}
The global convergence of several exact iterative methods that solve \eqref{orig_PG} has been explored, under the framework of Kurdyka–Lojasiewicz (KL) theory, in various additional literature including \cite{attouch2010proximal,attouch2013convergence,bolte2014proximal,razaviyayn2013unified,tseng2009coordinate}. Other work (see \cite{hu2021linear} and references therein) considered the linear convergence of non-exact algorithms with relaxations on the assumptions of KL theory, however, it is difficult to verify that the sequence generated by algorithm \ref{PG_alg} satisfies the relaxed assumptions stated in ~\cite{hu2021linear}. 
\end{remark}
\vspace{-0.5cm}
\section{Numerical Results for SVR Problem}
In this section, we present two numerical examples for the p-quasi-norm ADMM (pQN-ADMM) from algorithm \ref{alg:1} and the non-convex projection from algorithm \ref{alg:2}. For both examples, the pQN-ADMM algorithm result is compared with the $\ell_{1}$ objective function solution from MOSEK solver \cite{mosek}. The two examples include; i) Sparse signal reconstruction from noisy measurements, where the pQN-ADMM algorithm is also compared with another $\ell_{0.5}$ quasi-norm minimization based algorithm, named $\ell_{0.5}$-FL, described in \cite{FOUCART2009395}. ii) Binary classification using support vector machines (SVM).
\vspace{-0.3cm}
\subsection{Sparse Signal Reconstruction} \label{SSR}
Let $n=2^{10}$ and $m=n/4$, randomly construct the sparse binary matrix, $\textbf{M}\in \mathbb{R}^{m \times \frac{n}{2}}$, with a few number of ones in each column. The number of ones in each column of $\mathbf{M}$ is generated independently and randomly in the range of integers between $10$ and $20$, and their locations are randomly chosen independently for each column. Let $\mathbf{U}=[\mathbf{M},-\mathbf{M}]$, which is the vertical concatenation of the matrix $\textbf{M}$ and its negative. Following the same setup in \cite{22}, the column orthogonality in $\mathbf{U}$ is not satisfied. Let $\mathbf{x}_{\mathrm{opt}} \in \mathbb{R}^{n} $ be a reference signal with $\|\mathbf{x}_{\mathrm{opt}}\|_{0}=\lceil 0.2n \rceil$, where the non-zero locations are chosen uniformly at random with the values following a zero mean, unit variance Gaussian distribution. Let $\textbf{v}=\mathbf{U}\mathbf{x}_{\mathrm{opt}}+\textbf{n}$ be the allowable measurement, where $\textbf{n} \in \mathbb{R}^{m}$ is a Gaussian random vector with zero mean and co-variance matrix $\sigma^{2}\mathbf{I}_{m\times m}$, where $\mathbf{I}$ is the identity matrix. The sparse vector is reconstructed from $\textbf{v}$ by solving (\ref{motivProb}) with $\mathcal{V}=\{ \mathbf{x}: \|\mathbf{U}\mathbf{x}-\mathbf{v}\|/\|\mathbf{v}\|-\epsilon \leq 0 \}$.
\begin{figure} 
\begin{center}
\includegraphics[scale=0.5]
{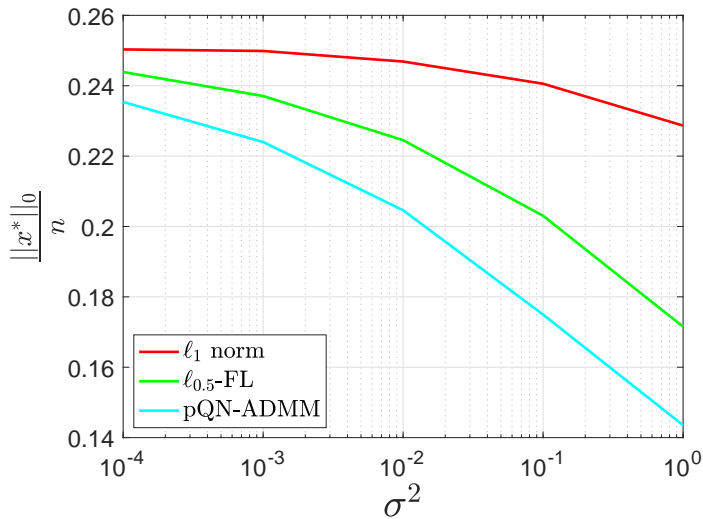}
\caption{Effect of noise variance on the sparsity of solutions obtained by pQN-ADMM algorithm, $\ell_{0.5}$-FL algorithm and $\ell_{1}$ norm minimization.} \label{Fig3}
\end{center}
\end{figure}
Figure \ref{Fig3} plots the relation between the sparsity level and the noise variance for $\ell_{1}$ norm minimization, $\ell_{0.5}$-FL quasi-norm and pQN-ADMM solutions. A threshold value of $10^{-6}$ was used where the threshold is a value below which the entry of the solution vector is considered to be zero. Depending on the noise variance $\sigma^{2}$, the value of $\epsilon$ was chosen to make the problem feasible. The reported result is the average of 100 independent random runs. It can be realized that pQN-ADMM algorithm produces a sparser solution than its counter baselines for different values of $\sigma^{2}$. On increasing $\sigma^{2}$, the sparsity level for all methods decreases. This is due to the increased scarcity of information on the original signal in the realization vector which makes the reconstruction process less accurate. \looseness=-1
\vspace{-0.3cm}
\subsection{Binary Classification} 
In this part, we build an email spam classifier based on support vector machines. We use a subset of the training set used in the SpamAssassin Public Corpus \cite{23}. Let $\{(\mathbf{u}_{j},v_{j})\}_{j \in [m]}$ be the training set of feature vectors $\mathbf{u}_{j} \in \{0,1\}^{n}$ with corresponding labels 
$v_{j} \in \{-1,1\}$ identifying whether the email is spam or not. We highlight the effectiveness of our method in designing an email spam detector using the least number of words. Following \cite{24}, we maintain a dictionary of $n=1899$ words. For a given email $j\in[m]$, the $i$th entry of $\mathbf{u}_{j}$ is 1 if word $\mathrm{w}_{i}, i \in [n]$ of the dictionary is in email $j$, and is 0 otherwise. We aim to build a linear classifier with the decision rule $\hat{v}=\mathrm{sign}(\mathbf{u}^{\top} \mathbf{x})$, where $\mathbf{u}$ is the feature vector of the email in question and $\mathbf{x}$ is a vector of the classifier coefficients with the first entry being the bias term. The main aim is to build a classifier that detects whether an email is a spam or not, using the least number of words from the dictionary and achieving a high training data accuracy. To achieve this objective, we solve (\ref{motivProb}) with $\mathcal{M}=\{\mathbf{x}: \frac{1}{m} \sum_{j \in [m]} \left( 1-v_{j}\mathbf{u}_{j}^{\top}\mathbf{x} \right)^{+}-\epsilon\leq 0$\}.

It can be clearly realized that the training set accuracy is controlled by $\epsilon$. Algorithm \ref{alg:1} was run for $p=0.5$, 2000 training emails and various values for $\epsilon$. For each value of $\epsilon$, the algorithm was terminated after 100 iterations and performance tested on 1000 emails. For comparison purpose, the problem was also solved with the $\ell_{1}$ norm convex relaxation under the same setup. In figure \ref{Fig4}, we plot the number of non-zero entries in the optimal classifier from both the pQN-ADMM and $\ell_{1}$ solutions vs different values of $\epsilon$.

We used a threshold of $10^{-4}$, where the threshold is defined as in section \ref{SSR}. It can be realized from figure \ref{Fig4} that the pQN-ADMM solution outperforms the $\ell_{1}$ in terms of the number of words used for legitimacy detection. When the value of $\epsilon$ increases, the number of required words decreases for both $\ell_{0.5}$ and $\ell_{1}$ problems. This outlines the trade-off between the sparsity level of the classifier and its accuracy, i.e., small values of $\epsilon$ enforces a low classification error in expense of a less sparse solution. The corresponding training and test set accuracies for the obtained classifiers are plotted in Fig. \ref{Fig5}. Both figures \ref{Fig4} and \ref{Fig5} depict the performance of the pQN-ADMM solution from algorithm \ref{alg:1} in terms of the sparsity level while maintaining nearly the same level of accuracy as the $\ell_{1}$ solution for both the training and test sets. 

\begin{figure}[t]
     \centering
     \begin{subfigure}[b]{0.4\textwidth}
         \centering
         \includegraphics[scale=0.3]{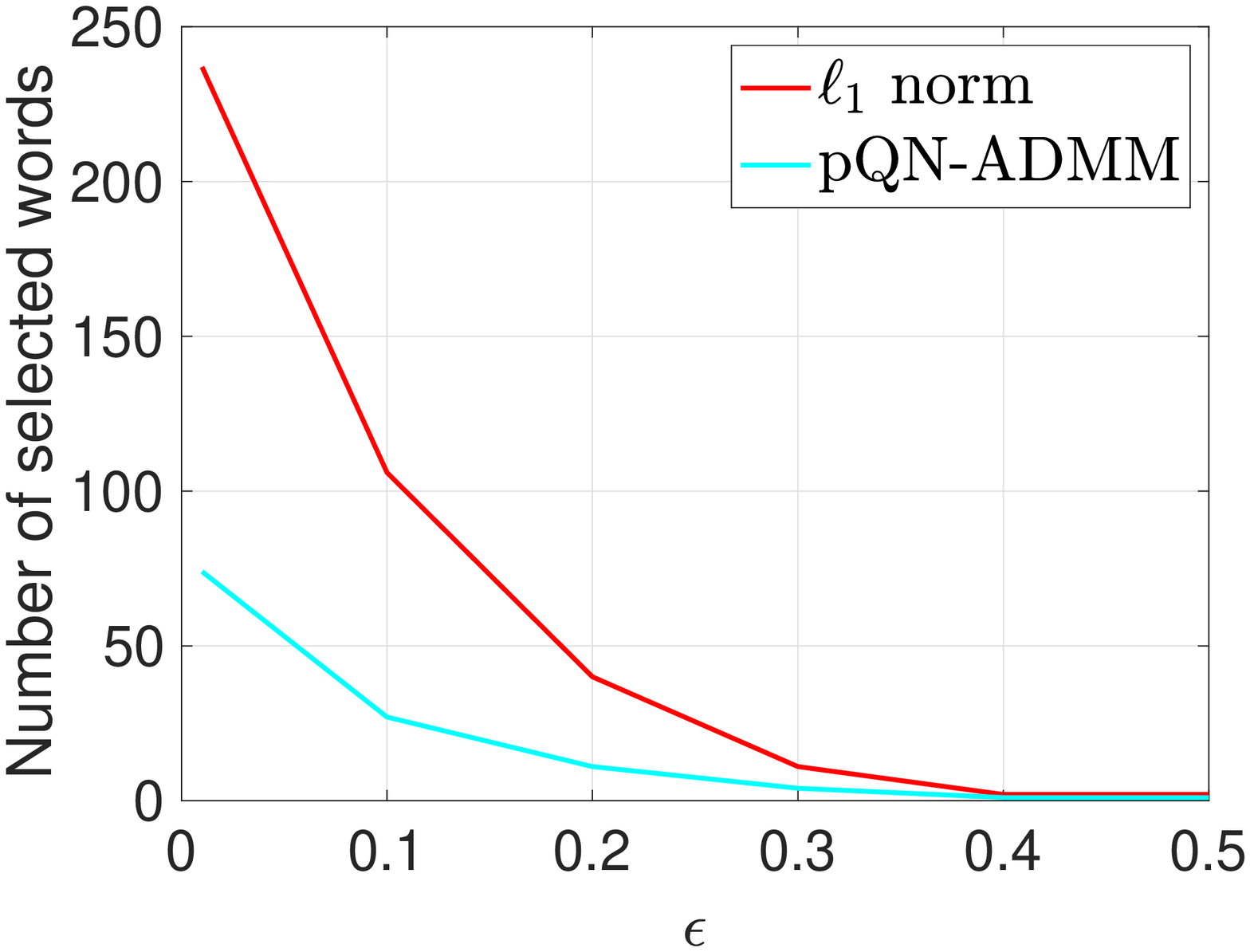}
         \caption{Number of words selected for classification versus $\epsilon$ for pQN-ADMM and $\ell_{1}$ norm.} 
         \label{Fig4}
     \end{subfigure}    
     \begin{subfigure}[b]{0.4\textwidth}
         \centering
         \includegraphics[scale=0.3]{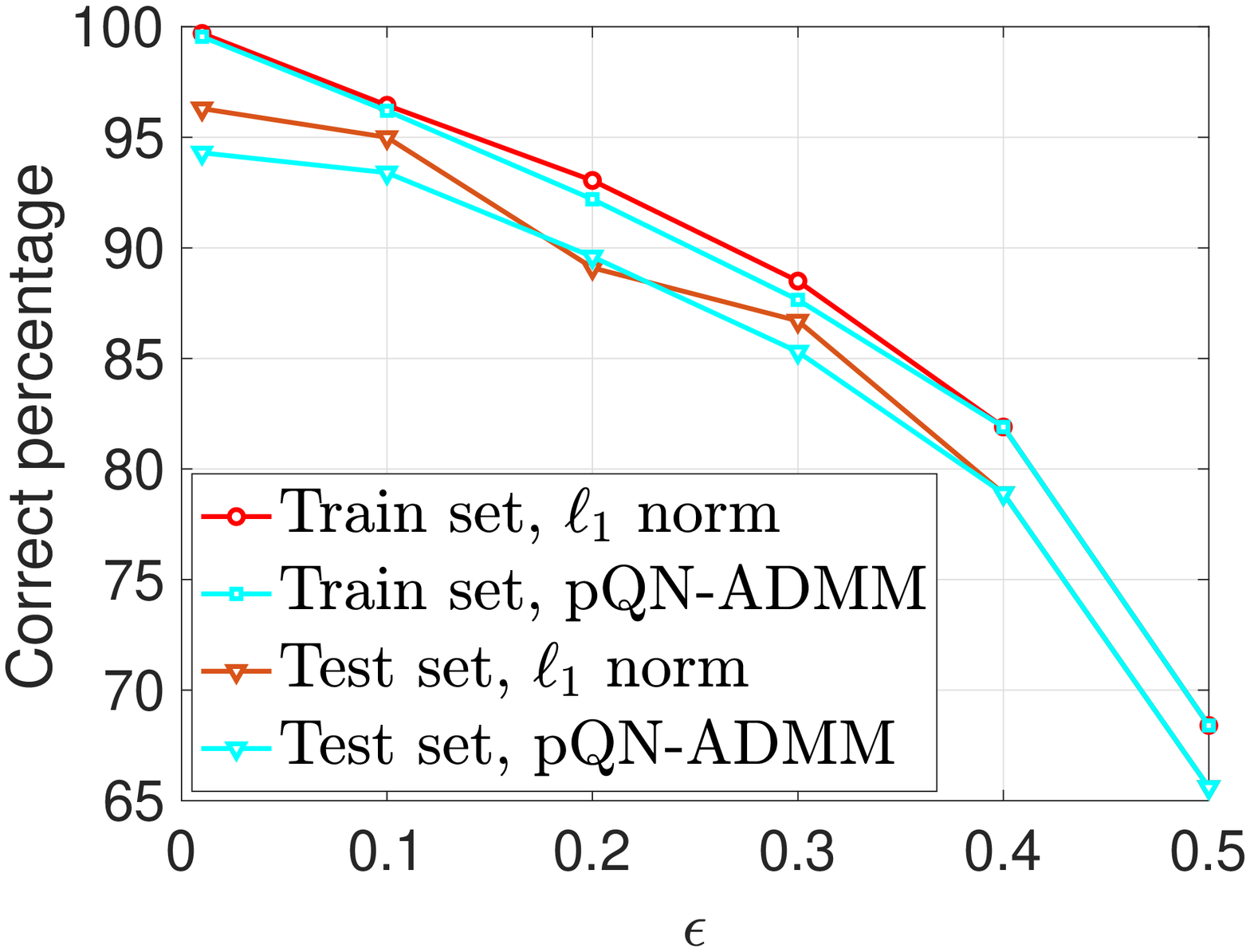}
         \caption{Training and test set accuracies versus $\epsilon$ for pQN-ADMM and $\ell_{1}$ norm.} 
         \label{Fig5}
     \end{subfigure}
     \caption{Binary classification numerical results.}
\end{figure}

\vspace{-0.3cm}
\section{Numerical Results for RMP Problem}
\subsection{Time domain system identification}
In this part, we apply the derived pQN-ADMM approach on a time domain system identification example. In that example, input is applied to randomly generated systems with a known order. Using the outputs corresponding to these systems, the minimum rank/order system is derived and results are compared to nuclear norm heuristic in \cite{4}. 

We consider a discrete time stable Single Input Single Output (SISO) system with an input $\mathbf{u} \in \mathbb{R}^{ \mathrm{T}}$, where $\mathrm{T}$ represents the number of input samples, i.e., input time span. We assume an impulse response of a fixed number of samples $n$. 
The corresponding system output is $\mathbf{y} \in \mathbb{R}^{m}$. However, we assume that only noisy realizations, $\mathbf{\mathbf{\hat{y}}}$, of the output can be considered, such that; $\mathbf{\hat{y}}\stackrel{\Delta}{=}\mathbf{y}+\mathbf{z}=\mathbf{h} \circledast \mathbf{u} +\mathbf{z}$
, where $\mathbf{h} \in \mathbb{R}^{ n}$ is the system's original impulse response, $\mathbf{z} \in \mathbb{R}^{m}$ is a random vector with entries drawn independently from samples of a uniform distribution on the range $[-0.25,0.25]$, i.e., $z_{i} \sim U[-0.25,0.25]$, while $\circledast$ denotes the convolution operator. From the window property of the convolution, $m=n+\mathrm{T}-1$. Assume that $u_{i}$, $h_{i}$ and $y_{i}$ are the $i$th components of the vectors $\mathbf{u}$, $\mathbf{h}$ and $\mathbf{y}$ respectively. The three components are related to each other by convolution through $y_{i}=\sum_{j=-\infty}^{\infty}h_{j}u_{i-j}$ which is a linear relation. Hence, let $\mathcal{T} \in \mathbb{R}^{m\times n}$ be the Toeplitz matrix formed by the input $\mathbf{u}$
, it can be easily seen that $\mathbf{h} \circledast \mathbf{u}=\mathbf{h} \mathcal{T}^\top$. Assume that $\mathbf{x}\in \mathbb{R}^{n}$ is an impulse response variable and let $\mathbf{X} \in \mathbb{R}^{n\times n}$ be a Hankel matrix formed by the entries of $\mathbf{x}$.
From \cite{6,15,16,17}, the minimum order time domain system identification problem can be formulated as,
\begin{subequations}  \label{15}
\begin{align} 
\min_{\mathbf{x},\mathbf{X}} \quad & \textbf{Rank}\mathbf{X}, \\ 
\textrm{s.t.} \quad & \mathbf{X}=Hankel(\mathbf{x}), \label{15-a}\\ 
& \left\lVert \mathbf{\mathbf{\hat{y}}}-\mathbf{x} \mathcal{T}^\top \right\rVert^{2}\leq \epsilon,  \label{15-b}
\end{align} 
\end{subequations}
(\ref{15-a}) ensures that $\mathbf{X}$ is a Hankel matrix and (\ref{15-b}) holds to make the result by applying the input, $\mathbf{u}$, to the optimal impulse response, $\mathbf{x}$, fit the available noisy data, $\mathbf{\hat{y}}$, in a non-trivial sense. Defining the convex set $\mathcal{C}\!\stackrel{\Delta}{=}\!\{\mathbf{X} \!\in \!\mathbb{R}^{n\times n}:\left\lVert \mathbf{\hat{y}}\!-\!h \mathcal{T}^\top \right\rVert^{2}\!\!\!-\!\epsilon \leq 0,  \mathbf{X}\!=\!Hankel(\mathbf{x})\}$, (\ref{15}) can be cast as, 
\begin{equation}  \label{16}
\begin{aligned}
\min_{\mathbf{x},\mathbf{X}} \quad & \rank(\mathbf{X}), \quad
\textrm{s.t.}  & \mathbf{X} \in \mathcal{C}, \\
\end{aligned}
\end{equation}
which is clearly identical to the problem in (\ref{1}). The problem was solved using the same pQN-ADMM approach discussed in section \ref{section 2}.

We let $\mathrm{T}\!=m\!=50$ and $n\!=\!40$. Note that $m<\mathrm{T}+n-1$, which is a reasonable assumption as in some practical applications, one is allowed only a specific window to realize the output. We consider the simulation for 10 different original system orders, i.e., $\eta\!=\!2\!:\!2\!:\!10$. An input vector, $\mathbf{u}$, is generated, where the elements of $\mathbf{u}$ are independent and follow a uniform distribution on the interval $[-5,5]$. For each $\eta$; 1) 50 random stable systems are generated using the command 'drss' in MATLAB. 2) The generated input is applied to each system to get the corresponding noisy output $\mathbf{\hat{y}}$. 3) Given the output $\mathbf{\hat{y}}$, the problem in (\ref{15}) is solved and the corresponding system's rank is calculated using singular value decomposition. 4) The results are averaged out to get the corresponding average rank to each original $\eta$. 

\begin{figure}[t]
\begin{center}
\includegraphics[scale=0.5]
{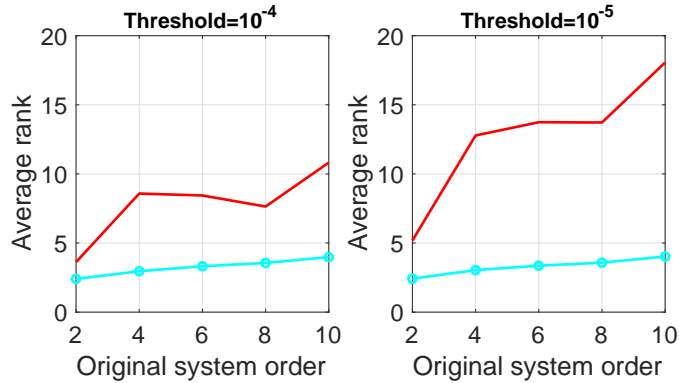}
\caption{Average rank vs original system order. Red and cyan colors are for the nuclear norm and pQN-ADMM algorithm respectively.} \label{fig_time_mean1}
\end{center}
\end{figure}

\begin{table} 
\begin{center}
\hspace*{0.5cm}
\hskip-0.25cm
\begin{tabular}{|p{3cm}|p{1cm}|p{1cm}|p{1cm}|}
 \cline{2-4}
 \multicolumn{1}{c}{} & \multicolumn{1}{|c|}{$\eta$=2} &\multicolumn{1}{c}{$\eta$=6}&\multicolumn{1}{|c|}{$\eta$=10}\\
 \hline
 \textbf{Nuclear norm}     &2.3907    & 6.6668& 7.2572\\
 \hline
 \textbf{pQN-ADMM}      &0.5292    & 0.9042& 1.0861\\
 \hline
\end{tabular}
\vspace{0.25cm}
\caption{\label{table1}Standard deviation for threshold=$10^{-4}$}
\end{center}
\end{table}
\begin{table} 
\begin{center}
\hspace*{0.5cm}
\hskip-0.25cm
\begin{tabular}{|p{3cm}||p{1cm}|p{1cm}|p{1cm}|}
 \cline{2-4}
 \multicolumn{1}{c|}{}  &\multicolumn{1}{|c|}{$\eta$=2}&\multicolumn{1}{|c|}{$\eta$=6}&\multicolumn{1}{|c|}{$\eta$=10}\\
 \hline
 \textbf{Nuclear norm}      &6.9877    & 11.2638&11.7854\\
 \hline
 \textbf{pQN-ADMM}      &0.5325     & 0.9113& 1.0861\\
 \hline
\end{tabular}
\vspace{0.25cm}
\caption{\label{table2}Standard deviation for threshold=$10^{-5}$}
\end{center}
\end{table}
Figure \ref{fig_time_mean1} shows the average rank for the the nuclear norm and pQN-ADMM heuristics. The results are for two different values of thresholds, where the threshold is defined as the value below which the singular value is considered to be zero. It can be realized that the introduced pQN-ADMM approach outperforms the nuclear norm one for both values of thresholds. Moreover, when the threshold value decreases from $10^{-4}$ to $10^{-5}$, the behavior of the pQN-ADMM remains the same. However, the average rank for the nuclear norm increases. This proves the robustness of the derived pQN-ADMM in comparison to the nuclear norm one. 
Tables \ref{table1} and \ref{table2} show the standard deviation of the algorithms. It can be seen that the standard deviation is the same for the pQN-ADMM when changing the threshold, however, it increases for the nuclear norm as the threshold value decreases.

\subsection{Matrix Completion Example}
In this section, we apply our algorithm (pQN-ADMM) to a matrix completion example and compare the result to the matrix iterative re-weighted least squares (MatrixIRLS) \cite{matrixIRLS1,matrixIRLS2}, truncated iterative re-weighted unconstrained Lq (tIRucLq) \cite{ref8} and iterative re-weighted least squares (sIRLS-p $\&$ IRLS-p) \cite{SIRLS-IRLS} algorithms. The matrix completion problem is a special case of the low rank minimization where a linear transform takes a few random entries of an ambiguous matrix $\mathbf{X}\!\in\! \mathbb{R}^{m\times n}$. Given only these entries, the goal is to approximate $\mathbf{X}$ and find the missing ones. The matrix completion problem with low rank recovery can be approximated by,
\begin{equation}  \label{comp_prob}
\begin{aligned}
\min_{\mathbf{X}} \quad & \left\lVert \mathbf{X} \right\rVert_{p}^{p}, \quad
\textrm{s.t.}  &  \left\lVert \mathcal{A}(\mathbf{X})-b \right\rVert\leq \epsilon , \\
\end{aligned}
\end{equation}
where $\mathcal{A}:\mathbb{R}^{m\times n} \to \mathbb{R}^{q}$ is a linear map with $q \ll mn$ and $b \in \mathbb{R}^{q}$. In order to apply the mentioned algorithms, the linear transform $\mathcal{A}(\mathbf{X})$ will be rewritten as $\mathrm{A} \mathrm{vec}(\mathbf{X})$, where $\mathrm{A} \in \mathbb{R}^{q\times mn}$ and $\mathrm{vec}(\mathbf{X}) \in \mathbb{R}^{mn}$ is a vector formed by stacking the columns of the matrix $\mathbf{X}$. 

A random matrix $\mathbf{M} \in \mathbb{R}^{m\times n}$ with rank $\mathrm{r}$ is created using the following method: 1) $\mathbf{M}=\mathbf{M}_{\mathrm{L}}\mathbf{M}_{\mathrm{R}}^{\top}$, where $\mathbf{M}_{\mathrm{L}} \in \mathbb{R}^{m\times \mathrm{r}}$ and $\mathbf{M}_{\mathrm{R}} \in \mathbb{R}^{n\times \mathrm{r}}$. 2) The entries of both $\mathbf{M}_{\mathrm{L}}$ and $\mathbf{M}_{\mathrm{R}}$ are i.i.d Gaussian random variables with zero mean and unit variance. 
Let $\hat{\mathbf{M}}=\mathbf{M}+\mathbf{Z}$, where $\mathbf{Z} \in \mathbb{R}^{m\times n}$ is a Gaussian noise with each entry being an i.i.d Gaussian random variable with zero mean and variance $\sigma^{2}$. The vector $b$ is then created by selecting random $q$ elements from $\mathrm{vec}(\hat{\mathbf{M}})$. Since $b=\mathrm{A}\mathrm{vec}(\hat{\mathbf{M}})$, one can easily construct the matrix $\mathrm{A}$ which is a sparse matrix where each row is composed of a value 1 at the index of the corresponding selected entry in the vector $b$ while the rest are zeros. We set $m=n=100$, $\mathrm{r}=5$ and $p=0.5$. Let $\mathrm{d}_{\mathrm{r}}=\mathrm{r}(m+n-\mathrm{r})$ denotes the dimension of the set of rank $\mathrm{r}$ matrices and define $\mathrm{s}=\frac{q}{mn}$ as the sampling ratio. We assume that $\mathrm{s}=0.195$ which yields to $q=1950$. It can be realized that $\frac{\mathrm{d}_{\mathrm{r}}}{q}<1$. 
We set $\sigma=0.1$ and let the algorithms terminate if a budget of 1000 iterations is 
reached. In order to compare the results from different algorithms, we consider the average of 50 runs for two measures: a) the relative Frobenius distance (RFD) to the matrix $\mathbf{M}$, b) the relative error to singular (REtS) values of $\mathbf{M}$. 

In figures \ref{fig:matrix_com_noisy} and \ref{fig:matrix_sings_noisy}, we report the average RFD and REtS values for all the algorithms. Despite that all the baselines are designed to exploit the specific structure of the matrix completion problem, described in \eqref{comp_prob}, while the proposed pQN-ADMM doesn't, it is competitive against them all. This in turns shows the effectiveness of the pQN-ADMM algorithm in solving the rank minimization problems without requiring any prior information about the structure of the associated convex set.



\begin{figure}[t]
     \centering
     \begin{subfigure}[b]{0.4\textwidth}
         \centering
         \includegraphics[scale=0.3]{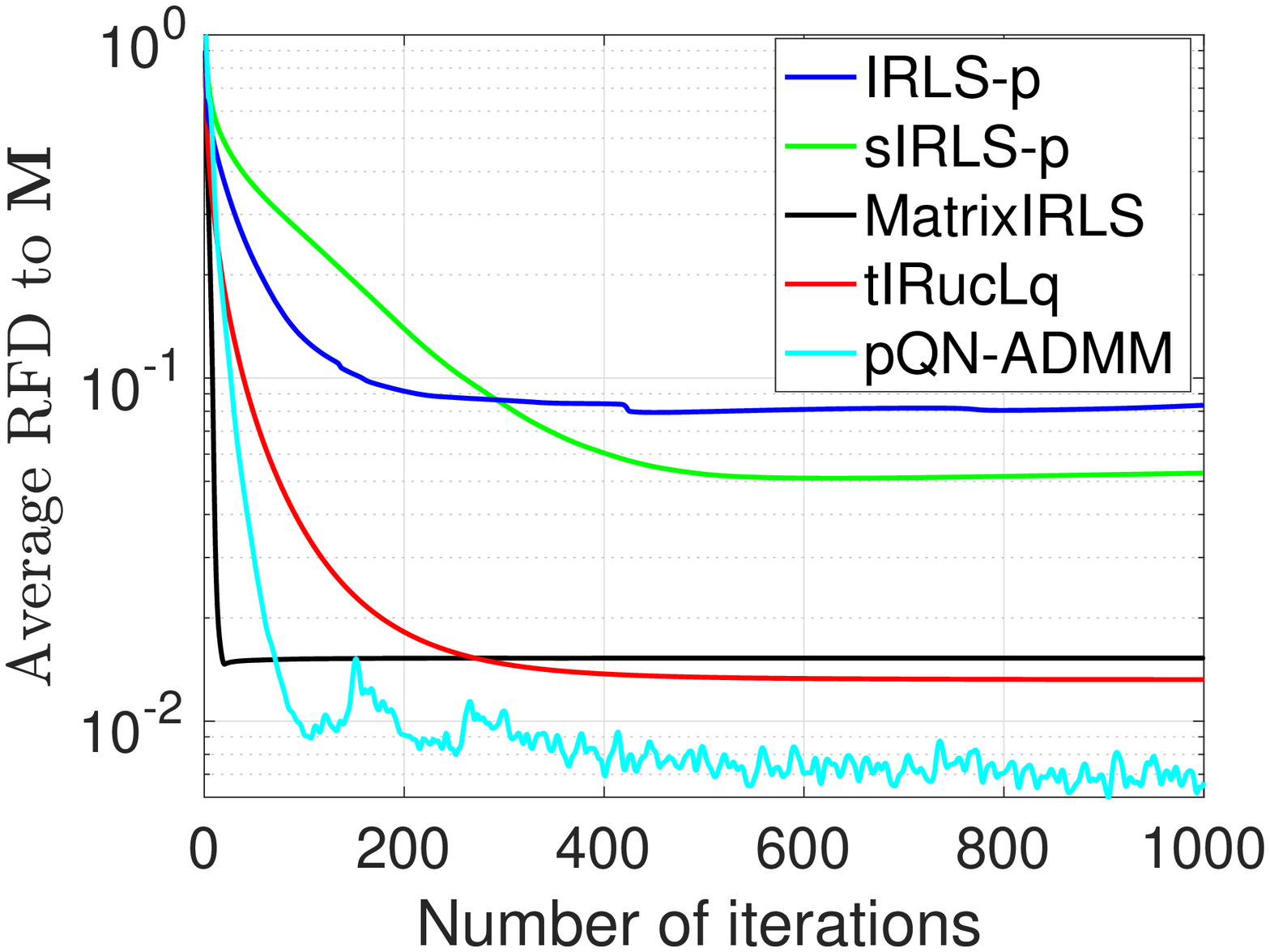}
         \caption{RFD to $\mathbf{M}$.} \label{fig:matrix_com_noisy}
     \end{subfigure}
     \begin{subfigure}[b]{0.4\textwidth}
         \centering
         \includegraphics[scale=0.3]{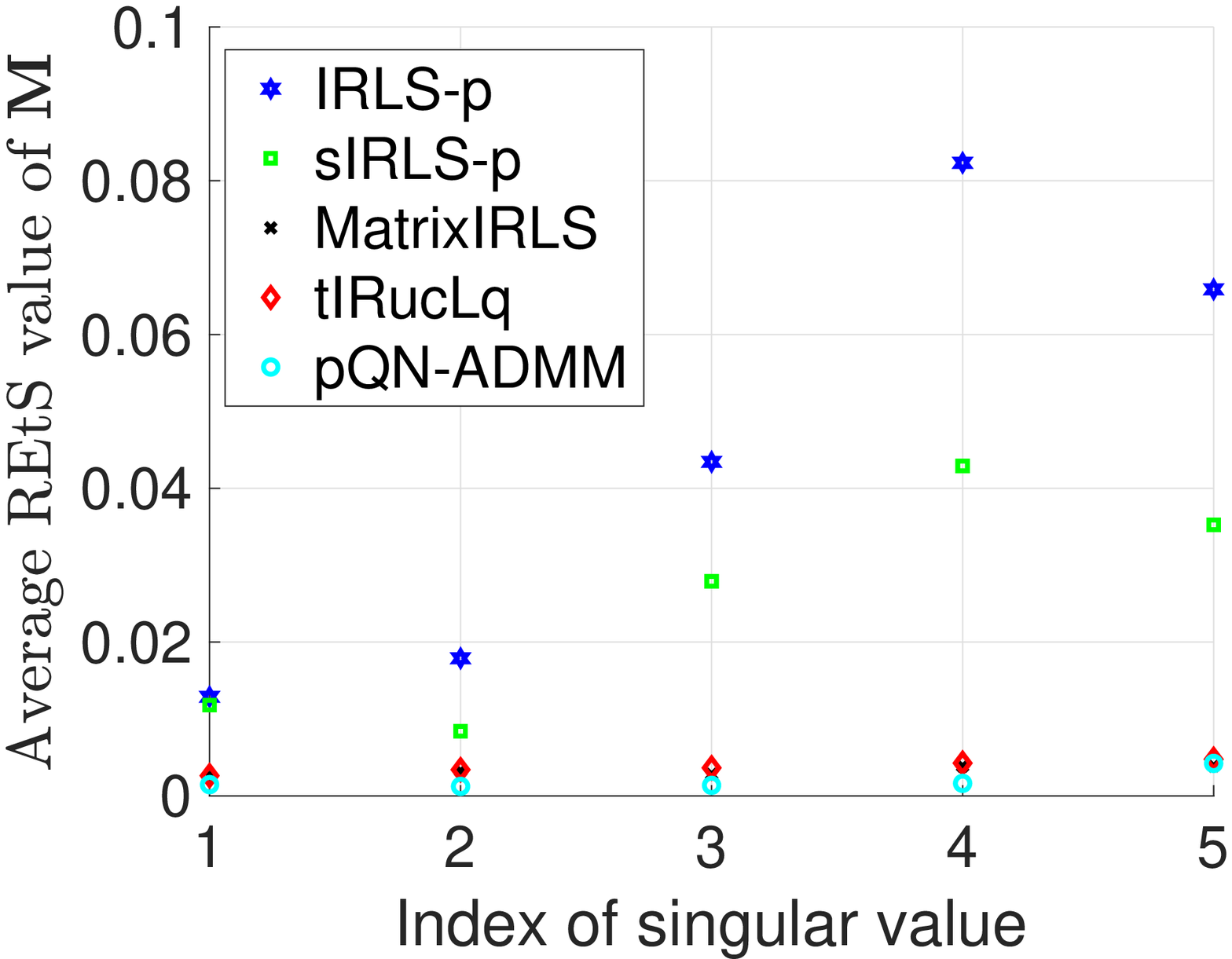}
         \caption{REtS values of $\mathbf{M}$.} 
         \label{fig:matrix_sings_noisy}
     \end{subfigure}
     \caption{The RFD and REtS average values.}
\end{figure}



\section{Numerical Results for the nonconvex Accelerated Proximal Gradient~(APG) Algorithm} 
In this subsection, we present numerical results for the APG method, 
displayed in Algorithm~\ref{PG_alg}. 
Following the same procedure in \cite{aybat2012first}, we first generate the target signal $\mathbf{x}^{*}$ through
\begin{equation} \label{target_sig_gen}
    x_{i}^{*}=
    \left\{\begin{array}{ll}
        \Theta_{i}^{(1)}10^{3\Theta_{i}^{(2)}}, &\forall~i\in\Lambda,\\
        0,  &\forall~i\in[n]\setminus\Lambda;
    \end{array}\right.
\end{equation}%
where the design parameters $\Lambda\subset [n]$, and $\Theta_i^{(1)}, \Theta_i^{(2)}$ for $i\in\Lambda$ are chosen as follows:
\begin{enumerate}
    \item the index set $\Lambda\subset[n]$ is constructed by selecting a subset of $[n]$ with cardinality 
    $s$ uniformly at random;
    \item $\{\Theta_{i}^{(1)}\}_{i\in\Lambda}$ are independent, identically distributed~(IID) Bernoulli random variables taking values $\pm 1$ with equal probability;
    \item $\{\Theta_{i}^{(2)}\}_{i\in\Lambda}$ are IID uniform $[0,1]$ random variables.
\end{enumerate}

The measurement matrix $A \in \mathbb{R}^{m\times n}$ is a partial Discrete Cosine Transform (DCT) matrix with rows corresponding to $m<n$ frequency, where these $m$ indices are chosen uniformly at random from $[n]$.
The noisy measurement vector $b \in \mathbb{R}^{m}$ is then set to be $b=A(\mathbf{x}^{*}+\epsilon_{1})+\epsilon_{2}$, where $\epsilon_{1}\sim \mathcal{N}(0,\sigma_1^2)$ and $\epsilon_{2}\sim \mathcal{N}(0,\sigma_2^2)$ are the input and realization noises. 

In our experiments, $n=4096$, $s=\ceil*{0.5m}$ and the PG algorithm memory to 5, i.e., $l=5$. Following the medium noise setup in \cite{hale2007fixed},  we set $\sigma_1=0.005$, $\sigma_2=0.001$.

For $f(\mathbf{x})=\left\lVert A\mathbf{x}-b\right\rVert^{2}$, we have $L=2\Vert A\Vert^2$.
We perform our experiment for various values of $m$, i.e., number of noisy measurements, and $\mu$, i.e., trade-off parameter, see~\eqref{PG}. For each $(m,\mu)$ selection, in order to capture the inherent statistical variation of the problem, we generate 20 random instances of the triplet $(\mathbf{x}^*, A, b)$ and each random instance is solved by Algorithm~\ref{PG_alg}. We reported the 
average performance. We terminated Algorithm~\ref{PG_alg} when the relative error between consecutive iterates satisfies $\left\lVert \mathbf{x}^{k}-\mathbf{x}^{k-1}\right\rVert/\left\lVert \mathbf{x}^{k-1}\right\rVert\leq 10^{-5}$ for the first time.

In our experiments, we compared solving \eqref{PG} for $p=0.5$ against $p=1$, i.e., against $\ell_1$-optimization for sparse recovery. On one hand, when $p=0.5$, i.e., for $\ell_{0.5}$ minimization, we solve \eqref{PG} using Algorithm~\ref{PG_alg}, called $\ell_{0.5}$ exact, and using the algorithm 2 of \cite{comp_paper_pg}, which we call $\ell_{0.5}$ approx. On the other hand, when $p=1$, $\ell_1$-minimization problem is a convex one and we adopt the FISTA algorithm of \cite{beck2009fast}. The solution is denoted by $\bar{\mathbf{x}}$ while the target signal, from \eqref{target_sig_gen}, by $\mathbf{x}^{*}$. In Algorithm~\ref{PG_alg}, $\mathbf{x}^{0}$ is set to a zero vector while $\mathbf{x}^{1}$ is the $\ell_{1}$ norm solution. 

Figures \ref{PG_Fig1} and \ref{PG_Fig3} highlight the relation between the average error and sparsity vs $\mu$ for different values of $n/m$. It can be realized that the average error (sparsity) decreases (increases) on increasing $\mu$. For small values of $\mu$, more weight is given to the loss function, which emphasizes the $\ell_{0}$ quasi-norm minimization, and hence the sparsity level, as in figure \ref{PG_Fig3}, is low. However, for high values of $\mu$, more weight is assigned to the minimization of the regularization term, which solves $\left\lVert A\mathbf{x}-b\right\rVert^{2}$, and hence the error decreases, as shown in figure \ref{PG_alg}, with a corresponding increase in the sparsity. It can be realized that the $\ell_{0.5}$ solutions always outperforms the $\ell_{1}$ one with very slight difference between the exact and the approximate ones.
\begin{figure}[t]
\begin{center}
\includegraphics[width=\columnwidth ]
{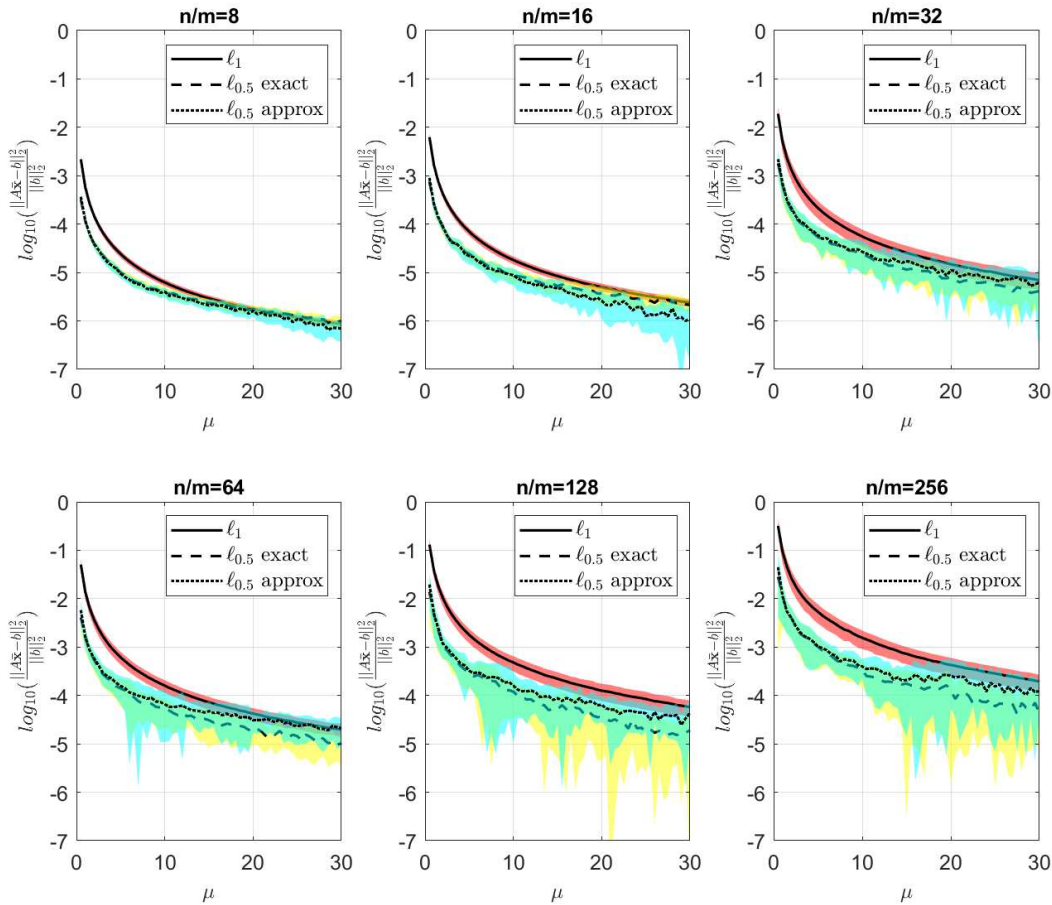}
\caption{Average error vs $\mu$ for different values of $n/m$. Yellow and cyan shades are the standard deviations for the exact and approximate $\ell_{0.5}$ quasi-norms respectively. } \label{PG_Fig1}
\end{center}
\vspace{-5mm}
\end{figure}
\begin{figure}[t]
\begin{center}
\includegraphics[width=\columnwidth]
{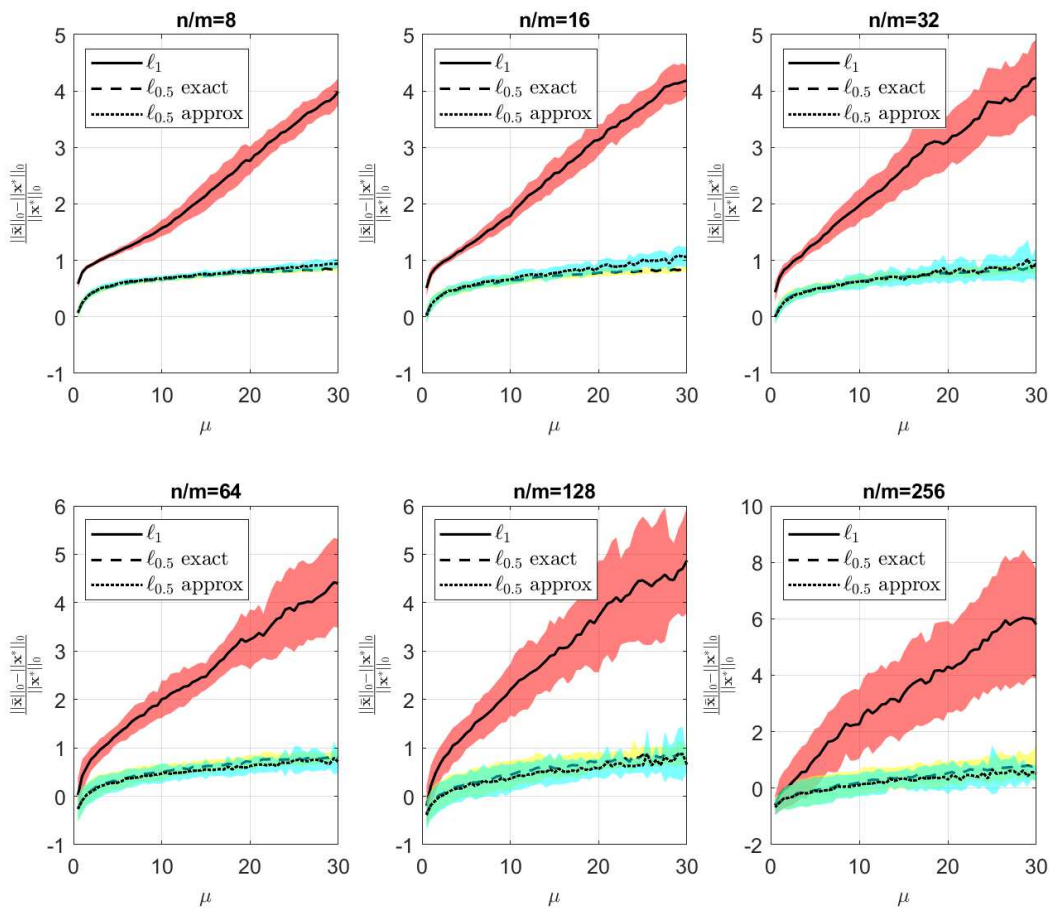}
\caption{Sparsity vs $\mu$ for different values of $n/m$. Yellow and cyan shades are the standard deviations for the exact and approximate $\ell_{0.5}$ quasi-norms respectively.} \label{PG_Fig3}
\end{center}
\vspace{-5mm}
\end{figure}

Figure \ref{PG_Fig4} highlights the statistics of the number of iterations used until convergence for both the $\ell_{0.5}$ exact and approximate algorithms. It can be realized that with a sufficient number of available realizations, $n/m=8$ and $n/m=16$, both algorithms approximately consume the same number of iterations. However, when the number of available realizations decreases, $n/m=32$ and higher, our exact proximal solution requires significantly less number of iterations to converge. This conclusion, along with figures \ref{PG_Fig1} and \ref{PG_Fig3} findings, indicates that our algorithm not only finds a similar solution to the approximate method, but also converges with a fewer number of iterations.
\begin{figure}[t]
\begin{center}
\includegraphics[width=\columnwidth]
{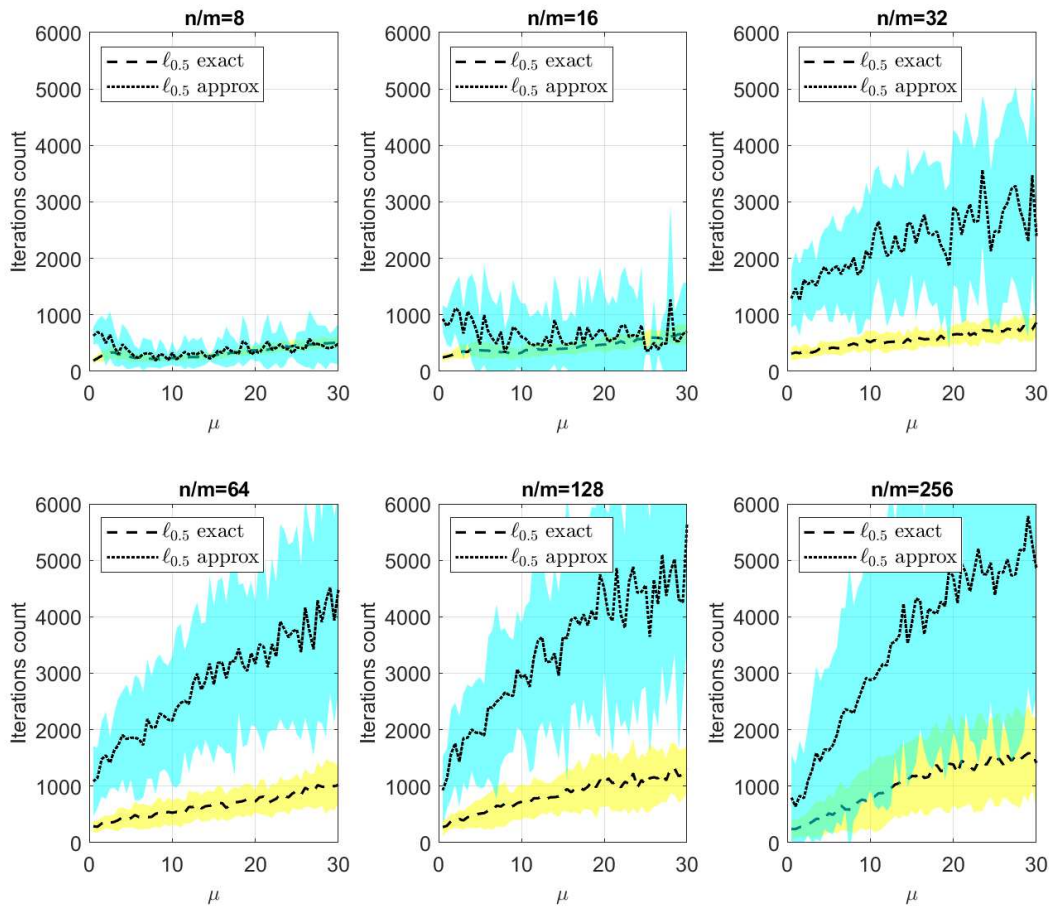}
\caption{Iterations count vs $\mu$ for different values of $n/m$.} \label{PG_Fig4}
\end{center}
\vspace{-5mm}
\end{figure}
\section{Conclusion}
In this study, we presented a non-convex ADMM algorithm (pQN-ADMM) to solve the $\ell_{p}$ norm minimization problem. The algorithm has a similar complexity to that of the $\ell_{1}$ minimization in addition to solving the roots of a polynomial for the non-convex projection. Our algorithm can also be considered as a general procedure for solving $\ell_{p}$ problems as no specific structure for the convex constraint set was assumed and a convex projection on that set was done for variables update. Applying sparse signal recover and binary classification examples, our method was found to outperform the $\ell_{1}$ minimization in terms of the sparsity of the generated solution. In addition, we studied the problem of solving a non-convex relaxation of RMPs using Schatten-p quasi-norm. This relaxation was shown to be the $\ell_{p}$ minimization of the singular values of the variable matrix and hence the primary developed algorithm could be used. Showing the numerical results, the pQN-ADMM was found to be less sensitive to the threshold decrease in time domain system identification problems. Additionally, the pQN-ADMM method was shown to be competitive against various other baselines when solving the matrix completion problem. 

\bibliography{Ref}
\bibliographystyle{ieeetr}

%

\end{document}